\documentclass[acmsmall,screen,nonacm]{acmart} %
\setcitestyle{nosort}

\let\oldmin=\min

\usepackage{mycommands}

\title{A flexible specification approach for verifying total correctness of fine-grained concurrent modules}

\author{Justus Fasse}
\orcid{0009-0006-7383-7866}
\affiliation{%
  \institution{KU Leuven}
  \department{Department of Computer Science}
  \streetaddress{Celestijnenlaan 200A box 2402}
  \city{Leuven}
  \postcode{3001}
  \country{Belgium}%
}
\email{justus.fasse@kuleuven.be}
\author{Bart Jacobs}
\orcid{0000-0002-3605-249X}
\affiliation{%
  \institution{KU Leuven}
  \department{Department of Computer Science}
  \streetaddress{Celestijnenlaan 200A box 2402}
  \city{Leuven}
  \postcode{3001}
  \country{Belgium}%
}
\email{bart.jacobs@kuleuven.be}

\begin{document}

\begin{abstract}
A well-established approach to proving progress properties such as deadlock-freedom and termination is to associate \emph{obligations} with threads.
For example, in most existing work the proof rule for lock acquisition prescribes a standard usage protocol by burdening the acquiring thread with an obligation to release the lock.
The fact that the obligation creation is hardcoded into the acquire operation, however, rules out non-standard clients e.g.\ where the release happens in a different thread.

We overcome this limitation by instead having the blocking operations take the obligation creation operations required for the specific client scenario as arguments.
We dub this simple instance of higher-order programming with auxiliary code \sassy.
To illustrate \sassy, we extend \heaplang, a simple, higher-order, concurrent programming language with erasable code and state.
The resulting language gets stuck if no progress is made.
Consequently, we can apply standard \emph{safety} separation logic to compositionally reason about \emph{termination} in a fine-grained concurrent setting.

We validated \sassy by developing (non-foundational) machine-checked proofs of representative locks---an unfair \spinlock (competitive succession), a fair \ticketlock (direct handoff succession) and the hierarchically constructed \cohortlock that is starvation-free if the underlying locks are starvation-free---against our specifications using an encoding of the approach in the \verifast program verifier for C and \java.
\end{abstract}

\maketitle

\section{Introduction}%
\label{sec:intro}%

\begin{figure}[h]
  \begin{tabular}{c||c||c}
\begin{lstlisting}
acquire lk;
(* $\color{black}\mathit{Critical~section}$ *)
f := false (* $\color{black}\mathit{handoff}$ *)
\end{lstlisting} &
\begin{lstlisting}
(* $\color{black}\text{\rmfamily\itshape confirm handoff}$ *)
if !f then exit;
(* $\color{black}\text{\rmfamily\itshape Critical~section}$ *)
release lk
\end{lstlisting} &
\begin{lstlisting}
acquire lk;
(* $\color{black}\mathit{Critical~section}$ *)
release lk
\end{lstlisting} \\
  \end{tabular}
  \caption[Motivating example]{The lock handoff between the left and middle thread is challenging for the modular verification of progress.
    This example reflects the key challenge also encountered when verifying a starvation-free \cohortlock.
    Here, we simplify the handoff by ``synchronizing'' via the exit statement, guarded by the dereference (\lstinline{!}) of flag \lstinline{f} initialized to true (not shown).
  }%
  \label{fig:motivating-example}
\end{figure}

\begin{figure}
  \begin{subfigure}[t]{.48\textwidth}
    \begin{tabular}{l}
\begin{lstlisting}
let acquire lk =
  while not
    /*⟨if !lk then $\cursivecomment[myghostcode]{Expect lk.ob}$
     else $\cursivecomment[myghostcode]{lk.ob := CreateOb CurrentThread}$;*/
     CAS lk false true/*⟩*/
  do () done
\end{lstlisting} \\
\begin{lstlisting}
let release lk =
  /*⟨$\cursivecomment[myghostcode]{DestroyOb lk.ob}$;*/ lk := false/*⟩*/
\end{lstlisting}
    \end{tabular}
    \caption{Classic obligation handling. The \spinlock module creates an obligation when the lock is acquired, which can only be destroyed by calling release.}%
    \label{fig:classic-spinlock}
  \end{subfigure}\hfill%
  \begin{subfigure}[t]{.48\textwidth}
    \begin{tabular}{l}
\begin{lstlisting}
let acquire lk /*η κ*/ =
  while not
    /*⟨if !lk then η ()
     else κ ();*/
     CAS lk false true/*⟩*/
  do () done
\end{lstlisting} \\
\begin{lstlisting}
let release lk /*κ*/ =
  /*⟨κ ();*/ lk := false/*⟩*/
\end{lstlisting}
    \end{tabular}
    \caption{Parameterized (\sassy) obligation handling.
      The client controls which obligations are created when acquiring the lock (\aalpha) and which can be relied upon to justify the blocking (\abeta).}%
    \label{fig:sassy-spinlock}
  \end{subfigure}
  \caption{\Spinlock implementation with two approaches to obligation-handling expressed as \textcolor{myghostcode}{auxiliary} (\textit{pseudo}) code.
    The angle brackets indicate the atomic execution of the contained expression to associate auxiliary code with the two cases of \lstinline{CAS x y z} (compare-and-set). If $\deref x = y$, \lstinline{CAS} assigns $y$ to $x$ and return true. Otherwise false is returned, leaving $x$ unchanged.
  }%
\label{fig:spinlock-classic-sassy}
\end{figure}

Consider the program in Figure~\ref{fig:motivating-example}.
We consider \lstinline{acquire} non-primitive e.g.\ implemented as a simple busy-waiting \spinlock (Figure~\ref{fig:spinlock-classic-sassy}).
Naturally, \lstinline{acquire} (the \spinlock module) is developed independently from specific client scenarios.
It is therefore oblivious to the non-standard usage pattern by the left and middle thread: the lock acquired by the left thread is released by the middle thread.
Nevertheless, termination of the right thread's \lstinline{acquire} crucially depends on the left and middle thread's correct synchronization (simulated via \lstinline{exit}).
One application of this pattern are \cohortlock{}s~\cite{DBLP:conf/ppopp/DiceMS12}.
In a \cohortlock the threads are partitioned into groups (cohorts) competing for a top-level lock.
When releasing, a thread prefers passing ownership of the top-level lock directly to another member of the same cohort without calling \lstinline{release}.
Such handoffs are synchronized via additional per-cohort locks (see Section~\ref{sec:cohortlock}).

We believe the \cohortlock is a good, self-contained, benchmark for compositional total correctness verification approaches.
Importantly, its starvation-freedom depends both on its bounded unfairness and the starvation-freedom of the underlying locks.

In this paper we propose a two-step verification approach for such programs.
First, based on the Ghost Signals constructs of \citet{DBLP:conf/cav/ReinhardJ20,Reinhard2021GhostSignalsTR}, we instrument the program-under-verification in a \emph{terminating} (fuel-based) language with \emph{obligations}-based auxiliary%
  \footnote{We distinguish logic-level ghost operations and erasable auxiliary code from now on.}
  constructs to soundly generate fuel, thereby supporting busy-waiting.
In this language, infinite fairly scheduled executions are impossible (they would run out of fuel and get stuck).
Consequently, proving \emph{safety} of the instrumented program establishes termination of the erased program under fair scheduling.

Classic modular obligations-based verification approaches for blocking programs~\cite{DBLP:conf/concur/Kobayashi06,DBLP:conf/esop/LeinoMS10,DBLP:conf/ecoop/BostromM15,DBLP:journals/toplas/JacobsBK18,DBLP:conf/cav/ReinhardJ20,Reinhard2021GhostSignalsTR} require a lock to be released by the thread that acquired it and therefore do not support the program in Figure~\ref{fig:motivating-example}.
Specifically, the acquiring thread is charged with an obligation which it can only get rid of by releasing the lock (Figure~\ref{fig:classic-spinlock}).
We lift this restriction by applying the higher-order programming specification approach of \citet{DBLP:conf/popl/JacobsP11} to the problem of modularly generating fuel.
This is possible since grounding the fuel-generation in the language sidesteps the open problem of reasoning about the fair scheduling assumption in a higher-order, step-indexed, logic (see Section~\ref{sec:related-work}).
Concretely, when an operation's progress depends on the client, we parameterize it with the auxiliary function $\abeta$ (e.g.\ Figure~\ref{fig:sassy-spinlock}).
While blocked on the client, invoking $\abeta$ must generate fuel.
The other auxiliary parameter \aalpha (Figure~\ref{fig:sassy-spinlock}) directly corresponds to \citet{DBLP:conf/popl/JacobsP11}'s auxiliary parameter and permits the manipulation of auxiliary state when the operation takes effect (cf.\ linearization point~\cite{DBLP:journals/toplas/HerlihyW90,DBLP:journals/pacmpl/BirkedalDGJST21}, see Section~\ref{sec:lat}) e.g.\ the moment the lock is acquired.
Equipped with $\abeta$ and $\aalpha$ it is quite simple to express the client-dependent progress argument (see Figure~\ref{fig:motivating-example-solution-sketch}).

\begin{figure}
  \begin{tabular}{c||c||c}
\begin{lstlisting}
acquire lk
  /*(λ. (* $\color{myghostcode}\mathit{expect~o2}$ *))*/
  /*(λ. (* $\color{myghostcode}\mathit{create~obligation}$
  $\color{myghostcode}\mathit{o1~in~\mathbf{middle}~thread}$ *))*/;
(* $\color{black}\mathit{Critical~section}$ *)
f := false (* $\color{black}\mathit{handoff}$ *)
$$
\end{lstlisting} &
\begin{lstlisting}
$$
$$
$$
if !f then exit;
(* $\color{black}\mathit{Critical~section}$ *)
release lk
  /*(λ. (* $\color{myghostcode}\mathit{destroy~o1}$ *))*/
\end{lstlisting} &
\begin{lstlisting}
acquire lk
  /*(λ. (* $\color{myghostcode}\mathit{expect~o1}$ *))*/
  /*(λ. (* $\color{myghostcode}\mathit{create~obligation}$
  $\color{myghostcode}\mathit{o2~in~current~thread}$ *))*/;
(* $\color{black}\mathit{Critical~section}$ *)
release lk
  /*(λ. (* $\color{myghostcode}\mathit{destroy~o2}$ *))*/
\end{lstlisting} \\
  \end{tabular}
  \caption[Solution sketch]{The pieces of auxiliary code passed to \lstinline{acquire} and \lstinline{release} encapsulate the client-dependent part of the progress proof: $\aalpha$ is executed at the point of lock acquisition/release; $\abeta$ justifies \lstinline{acquire}'s blocking. The left thread's blocking is justified by obligation \emph{o2} held by the right thread; the right thread's blocking is justified by obligation \emph{o1} held by the middle thread. Threads must destroy their obligations before they finish.}%
  \label{fig:motivating-example-solution-sketch}
\end{figure}

We gradually introduce \sassy, beginning with its operational perspective.
First, we introduce the fuel-based, terminating, language we use in the remainder of the paper (Section~\ref{sec:overview}).
Next, modular programming in this language allows us to complete the sketch of Figure~\ref{fig:motivating-example-solution-sketch} (Section~\ref{sec:motivating-example-revisited}).

Afterwards, we turn towards the issue of expressive, modular specifications.
We begin by recalling prior work on such specifications to compositionally prove \emph{partial} correctness of fine-grained concurrent data structures (Section~\ref{sec:lat}).
The next section introduce our \emph{total} correctness specifications for unfair blocking modules (Section~\ref{sec:tclat-unfair}).
Before defining our specifications for fair blocking modules (Section~\ref{sec:tclat-fair}) we discuss modular programming of fair modules in Section~\ref{sec:modular-programming-of-fair-modules}.
To illustrate these specifications, we apply them to several case studies: an unfair \spinlock, a fair \ticketlock and the hierarchically constructed \cohortlock (Section~\ref{sec:cohortlock}).
Section~\ref{sec:tool-support} briefly discusses the non-foundational mechanization of these case studies.
Lastly, we discuss related work (Section~\ref{sec:related-work}) and conclude the paper (Section~\ref{sec:conclusion}).

\section{A terminating language}%
\label{sec:overview}

\begin{figure}
  \input{./figures/heaplanglt-syntax.tex}
  \caption[\heaplanglt syntax]{
    Syntax of \heaplanglt. %
    Differences to \heaplang are highlighted in blue and are introduced in Section~\ref{sec:overview}.
    We write $\Lam \lvar. \expr \eqdef \MyRecE{\any}{\lvar}{\expr}$, $\Lam. \expr \eqdef \MyRecE{\any}{\any}{\expr}$, and $\Alloc~\expr \eqdef \AllocN(1,\expr)$.}%
  \label{fig:heaplanglt-syntax}
\end{figure}

\begin{figure}
  \begin{mathpar}
    \let\prophnil\empty %
    \let\inferC\inferH
    \input{./figures/heaplanglt-extensions.tex}
  \end{mathpar}
  \caption[\heaplanglt auxiliary operations]{\heaplanglt operational semantics. $\lev \prec O$ means $\forall (s, \lev') \in O.\;\lev < \lev'$.}%
  \label{fig:heaplanglt-extensions}%
  \label{fig:obligation-life-cycle-management}%
\end{figure}

\begin{figure}
  \input{./figures/state/heaplanglt-state-simplified.tex}%
  \caption[\heaplanglt's state]{\heaplanglt's state extends \heaplang's state with an auxiliary heap, a signals heap mapping allocated signal locations to their level $\lev \in \Levdom = \Val$ and state (unset or set), and three maps that map each live or finished thread's id $\theta \in \Theta = \mathbb{N}^+$ to the thread's finite multiset of obligations (redundantly including each signal's level), its finite multiset of call permissions (each qualified by a degree $\delta \in \Degdom = \Val$), and its finite multiset of expect permissions (each associated with a signal and a degree). 
    The semantics is parameterized by a well-founded order on $\Levdom$ and a well-founded order on $\Degdom$. In the \appendixorsupplement we discuss how these can be constructed modularly.}%
  \label{fig:heaplanglt-state}
\end{figure}

In this section we present our approach, based on Ghost Signals \cite{DBLP:conf/cav/ReinhardJ20,Reinhard2021GhostSignalsTR}, for extending a programming language with auxiliary constructs to obtain a language that always terminates under fair scheduling. The auxiliary constructs are erasable in the sense that if a program instrumented with these constructs does not get stuck, the original program under the original language's semantics also terminates under fair scheduling and also does not get stuck. We present our approach by applying it to \heaplang~\cite[\S~12]{iris-technical-reference}, a simple imperative language with unstructured concurrency (\lstinline{fork}) and a higher-order store, to obtain \heaplanglt.

The syntax, small-step semantics of our extensions, and state of \heaplanglt are shown in Figs.~\ref{fig:heaplanglt-syntax}, \ref{fig:heaplanglt-extensions}, and \ref{fig:heaplanglt-state} respectively.
The ``head step'' (\cite[\S~12]{iris-technical-reference}) relation $(e,\istate) \hstepi[]{\theta} (e',\istate',\nil)$ denotes that in state $\istate$, expression $e$, executed by thread $\theta$, can step to expression $e'$ with new state $\istate'$ and no forked threads ($\nil$).
Adopting the notation of \iris' technical documentation, we write $\istate : \stateHeap[\ell \gets \val]$ to designate the state identical to $\istate$ with the exception of $\istate.\stateHeap(\ell)$ which stores $\val$.
Moreover we use the in-place update notation $m[\iota \mathrel{\HLOp\!\!\hookleftarrow} \omega]$ to denote $m[\iota \gets (m(\iota) \HLOp \omega)]$ for a binary operation $\HLOp$.
We introduce the various constructs below.

\heaplang's sole facility for synchronization, and non-termination, are recursive function calls.
Thus the first step to convert \heaplang into our terminating language \heaplanglt is to consume some (trans-)finite resource (fuel) at every function call.%
  \footnote{We will often write while loops, which can easily be converted into recursive functions.}
Concretely we use call permissions~\cite{DBLP:journals/toplas/JacobsBK18} as fuel.
Call permissions are qualified by a \emph{degree} $\delta$ taken from a well-founded universe of degrees $\Degdom$.
At each function application a call permission at the dedicated, minimal degree $\degree[0]$ is consumed (Figure~\ref{fig:heaplanglt-extensions}, \ruleref{BetaS}).
If no call permission of degree $\degree[0]$ is available the program gets stuck.
The command $\langkw{lower}\spac \degree\spac \langkw{to}\spac n\spac \langkw{times}\spac \degreeA\spac \langkw{at}\spac \theta$ replaces a call permission at degree $\degree$ of thread $\theta$ with $n$ call permissions of of lower degree $\degreeA$ (\ruleref{LowerS}).
Notice that despite creating new call permissions at degree $\degreeA$, this command lowers the multiset of call permissions under multiset order~\cite{DBLP:journals/cacm/DershowitzM79,DBLP:journals/aaecc/Coupet-GrimalD06}.
We assume that every other degree is greater than $\degree[0]$ and omit the lowering to $\degree[0]$ if it is clear from context.
Additionally, we drop the $\langkw{at}$ clause if $\theta$ is the current thread.
Each thread has its own call permissions.
After a fork, the forkee starts off with a stock of call permissions equal to that of the forker at fork-time.

Call permissions scale to sequential and non-blocking fine-grained concurrent programs~\cite{DBLP:journals/toplas/JacobsBK18} but fail to address blocking concurrency.
The following program illustrates the problem.\\
\lstinline{let f = ref true in fork (while !f do () done); f := false}\\
After initializing the flag to true and forking, the main thread unsets the flag.
Meanwhile, the forkee busy-waits until the flag is unset.
Under any fair scheduler the main thread will eventually unset the flag, allowing the program to terminate.
Nonetheless, no amount of call permissions suffices to fuel this program.
Following \citet{Reinhard2021GhostSignalsTR} we address this issue by \emph{generating} fuel.

Naturally, fuel must not be generated arbitrarily. A thread may generate fuel only if some unset \emph{signal} exists whose setting it can legitimately expect. Specifically, $\CreateSignal[\theta]{\lev}$ creates an unset signal, returning its identifier $s$, and assigns the obligation to set $s$ to thread $\theta$ (Figure~\ref{fig:obligation-life-cycle-management}, \ruleref{NewSignalS}).
The obligation and signal are associated with the level $\lev \in \Levdom$ ($\lev$ is pronounced Cyrillic ``$l$''), where $\Levdom$ is a well-founded universe of levels that ensures acyclicity.
$\SetSignal[\theta]{s}$ sets $s$'s state to $true$, thus fulfilling (destroying) the associated obligation to set $s$ held by thread $\theta$. We omit $\theta$ if it is equal to the current thread.
If a thread still has obligations when it finishes, the program gets stuck (Figure~\ref{fig:obligation-life-cycle-management}, \ruleref{FinishS}).
Obligations can be passed to the forkee at thread creation time.
The division of obligations is expressed as a list of signal identifiers passed as extra \textcolor{myghostcode}{erasable} argument to $\langkw{fork}$: $\langkw{fork}\spac{\color{myghostcode}\overline{s}}\spac\expr$.

Finally, we need the operation that generates fuel: \lstinline{Expect}. Consider a simplified step rule:
\begin{mathpar}
\inferrule[UnsoundExpect]
  { \istate.\stateSignals(s) = (\lev,\myfalse) \\
    \lev \prec \istate.\stateObligations(\theta) }
  {(\Expect\spac \theta\spac s\spac \degree, \istate) \hstepi{\theta} (\TT, \istate : \stateCallPerms[\theta \updmsetunion \set{\degree}], \nil)}
\end{mathpar}
If $s$ is unset and its level is less than those of all obligations held by $\theta$, \lstinline{Expect θ s δ} generates a call permission of degree $\delta$ for $\theta$.
Once again, we omit $\theta$ if it is the current thread.

With \lstinline{Expect} in hand we can annotate the program with \textcolor{black}{auxiliary code} to fuel the busy-waiting.
\begin{lstlisting} 
let f = ref true in /*let s = NewSignal 0 in*/
fork /*∅*/ (while /*⟨if !f then Expect s 0;*/ !f/*⟩*/ do () done);
f := false; /*SetSignal s; Finish*/
\end{lstlisting}
The main thread pre-creates the signal $s$.
At fork-time, the two auxiliary arguments express that the main thread keeps the obligation to set $s$ and all its call permissions.
In order to \lstinline{Finish}%
  \footnote{The instrumented \lstinline{fork} implicitly inserts a \lstinline{/*Finish*/} at the end of the forkee.}%
  , it has to set $s$, which it does just after unsetting the flag.
As a result, the \lstinline{Expect s} in the forkee generates the call permissions necessary to justify the busy-waiting while the main thread is (legitimately) expected to set the flag.
We use atomic blocks---indicated via the angled brackets $\langle e \rangle$ (\ruleref{AtomicBlockS})---to execute \lstinline{Expect} in the same step the flag is read.
In our examples atomic blocks include exactly one real, atomic, expression.
Consequently they erase to that real atomic expression.

This fuel-generation scheme, however, is unsound; see the \appendixorsupplement for a counterexample.
Indeed, it generates fuel out of thin air.
We follow \citet{Reinhard2021GhostSignalsTR} and force fuel-generation to be \emph{prepaid}.
In addition to the acyclicity requirement already introduced, executing \lstinline{Expect θ s δ} requires thread $\theta$ to have an \emph{expect permission} for $s$ and $\delta$.%
  \footnote{``Wait permissions'' in \citet{Reinhard2021GhostSignalsTR}.}
$\CreateExpectPerm[\theta]{s}{\degree'}{\degree}$ spawns this expect permission by consuming a call permission of degree $\delta'$ ($\delta < \delta'$).
If no call permission of degree $\delta'$ is available to thread $\theta$ the program gets stuck.
We write \lstinline{Expect s} if $\theta$ is the current thread and there is a single expect permission for signal $s$.
Like call permissions, expect permissions are inherited by the forkee at fork-time (\ruleref{ForkS}).
With the additional expect permission requirement for \lstinline{Expect}, \heaplanglt is sound i.e.\ it does not admit infinite fair executions.
\begin{theorem}[Absence of infinite fair executions]\label{thm:no-infinite-fair-executions}
  A \heaplanglt program of the form $e; \Finish$ does not have infinite fair executions.
\end{theorem}
\begin{proof}
See \appendixorsupplement.
Due to the setup's similarity, the proof is similar to that of \citet{DBLP:conf/cav/ReinhardJ20,Reinhard2021GhostSignalsTR}.
\end{proof}

\begin{corollary}[Safe termination]
  A safe \heaplanglt program of the form $e; \Finish$ reduces to a value under fair scheduling.
  \label{thm:safe-termination}
\end{corollary}
\begin{proof}
  By Theorem~\ref{thm:no-infinite-fair-executions} we know that, under an arbitrary fair scheduler, \heaplanglt programs must either reduce to a value or get stuck.
  A safe \heaplanglt program cannot get stuck by definition.
\end{proof}

\section{Modular programming in \heaplanglt}%
\label{sec:motivating-example-revisited}
\label{sec:modular-programming-in-heaplanglt}

\begin{figure}
  \centering
  \begin{tabular}{c}
\begin{lstlisting}
let lk = new_lock () in let f = ref true in
/*let l_2 = ref 0 in let l_3 = ref 0 in lower δ_1 into 4 times δ_2; lower δ_2 into 4 times δ_3;*/
\end{lstlisting}
  \end{tabular}

  \begin{tabular}{r||l||l}
    \hspace{-.5cm}
\begin{lstlisting}
/*let started = ref false in*/
acquire lk
  /*(λ.(if not !started then
     NewExpectPerm !l_3 δ_2 δ_3;
     started := true);
     Expect !l_3)
  (λ.l_2 := NewSignal 2 л_1)*/;
(* $\color{black}\mathit{Critical~section}$ *)
f := false
$$
\end{lstlisting} &
\begin{lstlisting}
if !f then exit;
(* $\color{black}\mathit{Critical~section}$ *)
$$
$$
$$
$$
$$
$$
release lk
 /*(λ.SetSignal !l_2)*/
\end{lstlisting} &
\begin{lstlisting}
/*let started = ref false in */
acquire lk
  /*(λ.(if not !started then
     NewExpectPerm !l_2 δ_2 δ_3;
     started := true);
     Expect !l_2)
  (λ.l_3 := NewSignal л_1)*/;
(* $\color{black}\mathit{Critical~section}$ *)
release lk
  /*(λ.SetSignal !l_3)*/
\end{lstlisting} \\
  \end{tabular}
  \caption[Motivating example in \heaplanglt]{Safe \heaplanglt program that erases to the motivating example of Figure~\ref{fig:motivating-example}, using the \spinlock module of Fig.~\ref{fig:spinlock-classic-sassy}~(b).
    We elide the straightforward lowering of call permissions to $\degree[0]$ to justify function calls.}%
  \label{fig:motivating-example-folded}
\end{figure}

This section tackles consider the problem of modular programming in \heaplanglt: given a \heaplang program consisting of multiple modules, such as the program of Fig.~\ref{fig:motivating-example} built on top of a \spinlock module, how can we instrument it to obtain a safe \heaplanglt program, without breaking modularity?

The main challenge lies in instrumenting modules that offer operations whose termination depends on liveness properties of the client, such as lock modules where termination of an invocation of $\mathsf{acquire}$ depends on the client eventually calling $\mathsf{release}$ after each conflicting invocation of $\mathsf{acquire}$.

Two approaches may generally be adopted in this case, as illustrated in Fig.~\ref{fig:spinlock-classic-sassy}. One approach, illustrated in Fig.~\ref{fig:spinlock-classic-sassy}~(a), is for the module to burden the calling thread with sufficient obligations such that sufficient \lstinline{Expect} commands can be inserted into the module to fuel its busy-waiting loops. This yields the type of contract between module and client that has been enforced by most existing obligations-based modular verification approaches for blocking programs~\cite{DBLP:conf/concur/Kobayashi06,DBLP:conf/esop/LeinoMS10,DBLP:conf/ecoop/BostromM15,DBLP:journals/toplas/JacobsBK18,DBLP:conf/cav/ReinhardJ20,Reinhard2021GhostSignalsTR}. However, it imposes a particular usage pattern and generally rules out some subset of legitimate clients, such as the client of Fig.~\ref{fig:motivating-example} and the modular construction of \cohortlock{}s (Section~\ref{sec:cohortlock}).

Therefore, in this paper, we propose a more flexible and general approach, illustrated in Fig.~\ref{fig:spinlock-classic-sassy}~(b), where the job of generating fuel for the busy-waiting loops whose termination depends on the client is relegated to the client. Specifically, each operation whose termination depends on the client eventually performing some action takes, as an auxiliary argument $\abeta$, a piece of auxiliary code that must produce call permissions whenever it is invoked before the client has performed the action. In the case of $\mathsf{acquire}$, for example, $\abeta$ must produce a call permission whenever it is invoked while the client has not invoked a $\mathsf{release}$ operation after a conflicting $\mathsf{acquire}$ invocation.

This approach correspondingly also leaves it up to the client to create whatever obligations, in whatever threads, it needs to be able to do so. For this purpose, each module operation also takes, as an auxiliary argument $\aalpha$, a piece of auxiliary code that is invoked at the operation's linearization point \cite{DBLP:journals/toplas/HerlihyW90}.%
  \footnote{This revives \cite{DBLP:conf/popl/JacobsP11} which, for \emph{partial} correctness verification, has been superseded by \emph{viewshift}-based approaches (e.g.~\citet{DBLP:conf/popl/JungSSSTBD15}).}

Figure~\ref{fig:motivating-example-folded} shows how, given this parameterization of the \spinlock module, the client program of Figure~\ref{fig:motivating-example} can be instrumented.

First, the client sets up two auxiliary locations \lstinline{l_2} and \lstinline{l_3} that will store the signals whose obligations will be held by threads~2 (middle) and~3 (right), respectively.
Because creating an expect permission consumes a call permission, threads~1 and~3's $\abeta$ argument to $\mathsf{acquire}$ only creates one the first time it is invoked.
Thread~3 uses the $\aalpha$ argument to $\mathsf{acquire}$ to burden itself with the obligation to release upon acquiring the lock. 
In contrast, thread~1 burdens thread~2 instead, allowing it to finish without releasing the lock (which will happen in thread~2 instead).

\emph{Instantiating the degrees.}
The example client program is parameterized by the degree $\delta_3$ of the call permission expected by the lock module's $\mathsf{acquire}$ implementation, both when it is invoked and after it invokes its $\abeta$ argument, as well as by degrees $\delta_1$ and $\delta_2$ such that $\delta_3 < \delta_2 < \delta_1$.
In the \appendixorsupplement we discuss how to pick these degrees modularly.

\section{Background: partial correctness of possibly diverging programs}%
\label{sec:lat}

In this paper, we propose a two-step approach to modularly verify termination of busy-waiting programs: first, modularly instrument the program with auxiliary constructs from a terminating language such that safety of the instrumented program implies termination under fair scheduling of the original program; second, modularly verify safety of the instrumented program using any existing logic for partial correctness verification of higher-order fine-grained concurrent programs.

\begin{figure}
\begin{subfigure}[t]{.3\textwidth}
\begin{tabular}{c}
\begin{lstlisting}
let x = ref 1 in
let y = ref 1 in
fork (1 / !y; y := 0);
1 / !x; x := 0
$$
$$
\end{lstlisting}
\end{tabular}
\caption{The threads work with disjoint state.}%
\label{fig:safety-seplogic-independent}
\end{subfigure}%
\hfill
\begin{subfigure}[t]{.3\textwidth}
\begin{tabular}{c}
\begin{lstlisting}
let x = ref 1 in
fork (
  while not CAS x 0 1
  do () done;
  1 / !x; x := 0);
1 / !x; x := 0
\end{lstlisting} 
\end{tabular}
\caption{Location $x$ is concurrently accessed by both threads.
  The forkee's division by $x$'s content is guarded by the busy-waiting \lstinline{CAS}-loop}%
\label{fig:safety-seplogic-inlined}
\end{subfigure}%
\hfill
\begin{subfigure}[t]{.3\textwidth}
\begin{tabular}{c}
\begin{lstlisting}
let x = ref 1 in
fork (
  acquire x;
  1 / !x; release x);
1 / !x; release x
$$
\end{lstlisting}
\end{tabular}
\caption{Forker and forkee synchronize the division by $x$'s content using a \spinlock.}%
\label{fig:safety-seplogic-spinlock}
\end{subfigure}
\caption{Safe programs dividing by the value of a reference.}%
\label{fig:safety-seplogic}
\end{figure}

We show how to do the latter in \iris~
\cite{DBLP:conf/popl/JungSSSTBD15,DBLP:conf/icfp/0002KBD16,DBLP:journals/jfp/JungKJBBD18,DBLP:conf/esop/Krebbers0BJDB17,DBLP:journals/pacmpl/SpiesGTJKBD22}, a state-of-the-art logic in this category.
First, this section recalls key notion of \iris to verify \emph{partial} correctness of uninstrumented fine-grained concurrent programs in \iris.
Beginning with monolithic programs, this section builds up to \citet{DBLP:conf/popl/JungSSSTBD15}'s \emph{logically atomic triples} specification formalism for fine-grained concurrent modules.
The next section then shows how to adapt this approach to verify instrumented programs following the programming pattern of Section~\ref{sec:modular-programming-in-heaplanglt}, yielding \emph{total correctness logically atomic triples}.

\subsection{Separation logic}

Consider the three concurrent programs in Figure~\ref{fig:safety-seplogic}.
We will use \iris to prove safety/partial correctness for them; in particular, that they do not divide by zero.
\iris is a separation logic \cite{reynolds2000intuitionistic,DBLP:conf/csl/OHearnRY01,DBLP:journals/cacm/OHearn19}, a Hoare-style~\cite{DBLP:journals/cacm/Hoare69} logic where assertions assert \emph{ownership} of \emph{resources}, in addition to asserting the truth of logical propositions.

\begin{figure}
\begin{subfigure}{.3\textwidth}
\begin{tabular}{c}
\begin{lstlisting}
$\{ \TRUE \}$
let x = ref 1 in
$\{ x \mapsto 1 \}$
let y = ref 1 in
$\{ x \mapsto 1 \ast y \mapsto 1 \}$
fork (
  $\{ y \mapsto 1 \}$
  1 / !y;
  $\{ y \mapsto 1 \}$
  y := 0
  $\{ y \mapsto 0 \}$
);
$\{ x \mapsto 1\}$
1 / !x;
$\{ x \mapsto 1\}$
x := 0
$\{ x \mapsto 0\}$
\end{lstlisting}
\end{tabular}
\caption{Local reasoning about disjoint state.}%
\label{fig:safety-proof-outline-independent}
\end{subfigure}
\begin{subfigure}{.65\textwidth}
\begin{mathpar}
  \inferH{hoare-conseq}{P \Rrightarrow P' \\ \hoare{P'}{e}{Q'} \\ Q' \Rrightarrow Q}{\hoare{P}{e}{Q}}

  \axiomH{points-to-exclusivity}{\ell \mapsto \_ \ast \ell \mapsto \_ \implies \FALSE}

  \axiomH{hoare-alloc}{\hoare{\TRUE}{\Alloc\spac v}{\ell.\, \ell \mapsto v}}

  \axiomH{hoare-load'}{\{ \ell \mapsto v \}~\deref\spac\ell~\{w.\, v = w \ast \ell \mapsto v \}}

  \axiomH{hoare-store'}{\{ \ell \mapsto v \}~\ell\spac\coloneq\spac w~\{ \ell \mapsto w \}}

  \inferH{hoare-fork}{\{ P \}~e~\{ \TRUE \}}{\{ P \}~\langkw{fork}\spac e~\{ \TRUE \}}

  \inferH{hoare-frame}{\{ P \}~\expr~\{ v.\, Q \}}{\{ P \ast R \}~\expr~\{ v.\, Q \ast R\}}

  \axiomH{hoare-def'}{\hoare{P}{\expr}{Q} \eqdef \square (P \wand \wpre{\expr}{Q})}

\end{mathpar}
\caption{Selected reasoning rules.}%
\label{fig:seplogic-basic-rules}
\end{subfigure}
\caption{Separation logic enables local reasoning. The persistence modality $\square P$ expresses that $P$ does not assert ownership. }
\end{figure}

Figure~\ref{fig:safety-proof-outline-independent} shows an \iris proof outline for the program in Figure~\ref{fig:safety-seplogic-independent}, constructed using the proof rules shown in Figure~\ref{fig:seplogic-basic-rules}. The \emph{points-to assertion} $\ell \mapsto v$ asserts exclusive ownership of the memory cell at location $\ell$ and that it currently stores the value $v$; the \emph{separating conjunction} $P * Q$ asserts ownership of \emph{separate} resources $r_1$ and $r_2$ such that $P$ holds for $r_1$ and $Q$ holds for $r_2$. The proof of the $\langkw{fork}$ command uses \ruleref{hoare-frame} to \emph{frame off} $x \mapsto 1$ and \ruleref{hoare-conseq}, where, as we will see, the \emph{viewshift} operator $\Rrightarrow$ generalizes implication.

\iris' Hoare triples are not primitive; rather, they are encoded in terms of its weakest precondition $\textlog{wp}$~\cite{DBLP:journals/cacm/Dijkstra75}, separating implication (also known as \emph{magic wand}) $\wand$~\cite{DBLP:conf/popl/IshtiaqO01} and persistence modality $\square$~\cite{DBLP:conf/popl/JungSSSTBD15,DBLP:journals/jfp/JungKJBBD18}; see \ruleref{hoare-def'}. $\wpre{e}{Q}$ asserts ownership of sufficient resources to ensure expression $e$ executes safely and satisfies postcondition $Q$; $P \wand Q$ asserts ownership of some resources that, composed with any \emph{separate} resources for which $P$ holds, satisfy $Q$: $P * (P \wand Q) \proves Q$; $\square P$ asserts that $P$ holds \emph{persistently}: $\square P \proves \square P * \square P$.

\subsection{Sharing resources}\label{sec:invariants}

\begin{figure}
\raggedright
\begin{subfigure}{.55\textwidth}
\begin{tabular}{c}
\begin{lstlisting}
$\{ \TRUE \}$
let x = ref 1 in
$\{ x \mapsto 1 \}$
$\{ x \mapsto 1 \ast R \}~\textrm{where}~R \ast R \implies \FALSE$
$ \{ R \ast \knowInv{\namesp}{I} \}~\textrm{where}~I \eqdef \exists n.\, x \mapsto n \ast (n = 1 \lor (n = 0 \ast R))$
fork (
  $\{ \knowInv{\namesp}{I} \}$
  while not
    $\{ (x \mapsto 1) \lor (x \mapsto 0 \ast R) \}$
    CAS x 0 1
    $\{b.\, (b = \mytrue \ast x \mapsto 1 \ast R) \lor (b = \myfalse \ast x \mapsto 1) \}$
  do () done;
  $\{ R \ast \knowInv{\namesp}{I} \}$
  1 / !x; x := 0 (* $\cursivecomment{see main thread}$ *)
);
$\{ R \ast \knowInv{\namesp}{I} \}$
let v =
  $\{ R \ast x \mapsto 1 \}$
  !x
  $\{w.\, w = 1 \ast R \ast x \mapsto 1 \}$
in
$\{ v = 1 \ast R \ast \knowInv{\namesp}{I} \}$
1 / v;
$\{ R \ast \knowInv{\namesp}{I} \}$
  $\{ R \ast x \mapsto 1 \}$
  x := 0
  $\{ R \ast x \mapsto 0 \}$
$\{ \knowInv{\namesp}{I} \}$
\end{lstlisting}
\end{tabular}
\caption{Concurrent reads of, and writes to, \lstinline{x}.}%
\label{fig:safety-proof-inlined}
\end{subfigure}%
\hfill
\begin{subfigure}{.44\textwidth}
\begin{mathpar}
  \axiomH{hoare-def}{\hoare{P}{\expr}{Q}[\mask] \eqdef \square (P \wand \wpre{\expr}[\mask]{Q})}

  \inferH{hoare-csq}
    { \prop \vs[\mask] \prop' \\
      \hoare{\prop'}{\expr}{\propB'}[\mask] \\
      \propB' \vs[\mask] \propB}
    {\hoare{\prop}{\expr}{\propB}[\mask]}

  \axiomH{inv-alloc'}{\prop \proves \pvs \knowInv{\namesp}{\prop}}

  \inferH{hoare-inv}
    { 
      \knowInv{\namesp}{I}\\
      \atomic(\expr) \\
      \namesp \subseteq \mask \\
      \hoare{\later I \ast P}{\expr}{\later P \ast Q}[\mask\setminus\namesp]
    }
    {\hoare{P}{\expr}{Q}[\mask]}

  \axiomH{hoare-load}{\{ \ell \mapsto v \}~\deref\spac\ell~\{w.\, v = w \ast \ell \mapsto v \}_{\emptyset}}

  \axiomH{hoare-store}{\{ \ell \mapsto v \}~\ell\spac\coloneq\spac w~\{ \ell \mapsto w \}_{\emptyset}}

  \axiomH{hoare-cas-succ}{\hoare{\ell \mapsto v}{\CAS\spac\ell\spac v\spac w}{b.\, b = \mytrue \ast \ell \mapsto w}[\emptyset]}

  \inferH{hoare-cas-fail}{v \neq v'}{\hoare{\ell \mapsto v}{\CAS\spac\ell\spac v'\spac w}{b.\, b = \myfalse \ast \ell \mapsto v}[\emptyset]}

\end{mathpar}
\caption{Selected \iris rules~\cite{DBLP:journals/jfp/JungKJBBD18,iris-technical-reference}.}%
\label{fig:selected-concurrency-rules}
\end{subfigure}
\caption{The location \lstinline{x} is shared between the threads.
  The load, store, and compare-and-swap commands are atomic.
  }
\end{figure}

Figure~\ref{fig:safety-proof-inlined} shows a proof outline for the program of Figure~\ref{fig:safety-seplogic-inlined}, using the proof rules shown in Figure~\ref{fig:selected-concurrency-rules}. In this program, both threads read from and write to the same location $x$. To support the sharing of resources among threads, \iris has \emph{invariants}. In Iris, all resources are owned either by threads or by invariants. Creating an invariant $\knowInv{\namesp}{P}$ in a \emph{namespace} $\namesp$ transfers ownership of some resources satisfying $P$ from the current thread to the invariant: $P \Rrightarrow \knowInv{\namesp}{P}$.\footnote{The \emph{viewshift} $P \Rrightarrow Q$ states that resources satisfying $P$ can be transformed into resources satisfying $Q$ by updating \emph{ghost state}. \iris's invariants are constructed from (higher-order) ghost state.} The knowledge of the existence of an invariant is persistent: $\knowInv{\namesp}{P} \vdash \knowInv{\namesp}{P} * \knowInv{\namesp}{P}$, and can therefore be shared freely among threads.

A thread can \emph{open} an invariant; this temporarily transfers ownership of the invariant's resources to the thread. However, a thread must \emph{close} all invariants it opened before the end of the current step of execution. Also, it must not open an invariant that is already open. Both constraints are enforced by tracking the current \emph{mask}, the set of namespaces\footnote{Actually, namespaces and masks are both sets of \emph{invariant names}; if the current mask is $\mask$, an invariant in namespace $\namesp$ can be opened if $\namesp \subseteq \mask$.} that can currently be opened. Specifically, the mask $\mask$ in $\hoare{P}{e}{Q}[\mask]$ means that, if $e$ takes at most one step, invariants in mask $\mask$ may be opened to prove the triple; see \ruleref{hoare-inv}.\footnote{In order to be able to assign meaning to higher-order ghost state such as invariants, \iris is \emph{step-indexed}: \iris propositions are predicates over a \emph{resource} and a \emph{step index} $n$, i.e.~the number of steps left in the partial execution trace under consideration. $\later P$ asserts that $P$ holds \emph{later}, i.e.~after another step of execution. For points-to assertions and other first-order assertions $P$, which are \emph{timeless}, the modality can be ignored: $\later P \Rrightarrow P$.} If a mask is omitted, the default mask $\top$, which includes all namespaces, is assumed.

Central to the proof outline in Figure~\ref{fig:safety-proof-inlined} is the purely ghost, exclusive, resource $R$%
  \footnote{For example $\ownGhost{\gname}{\exinj ()}$~\cite[\S~4.6]{iris-technical-reference}. Specifically, we have $\TRUE \vs \exists \gamma.\,\ownGhost{\gname}{\exinj ()}$.}, initially owned by the main thread and, through the invariant, providing exclusive permission to set $x$ from 1 to 0.

\emph{Fancy updates.}
In fact, rule \ruleref{hoare-inv} is a special case of rule \ruleref{hoare-atomic} (Figure~\ref{fig:inv-detailed-rules}). The \emph{fancy update} $\pvs[\mask_1][\mask_2] P$ \cite[\S~7.2]{DBLP:journals/jfp/JungKJBBD18} asserts ownership of sufficient resources to enable changing the current mask from $\mask_1$ to $\mask_2$ (which involves closing the invariants in $\mask_2 \setminus \mask_1$ and opening the invariants in $\mask_1 \setminus \mask_2$) and otherwise updating ghost state to obtain resources satisfying $P$. In particular, if $\namesp \subseteq \mask$, we have $\knowInv{\namesp}{I} \proves \pvs[\mask][\mask\setminus\namesp] \later I * (\later I \vsW[\mask\setminus\namesp][\mask] \TRUE)$ (\textsc{inv-open}), where \emph{linear viewshift} $P \vsW[\mask_1][\mask_2] Q$ is a shorthand for $P \wand \pvs[\mask_1][\mask_2] Q$.

\subsection{Fine-grained concurrent functions} %

Figure~\ref{fig:safety-proof-deref} shows a proof outline for the program from Figure~\ref{fig:safety-seplogic-inlined} with $\deref x$ replaced by $\texttt{deref}\spac x$, defined as \lstinline{let deref l = !l}.
Unlike $\deref\ell$, $\texttt{deref}\spac\ell$ is not atomic since the function application itself takes a step too.
The proof of $\texttt{deref}$ uses \ruleref{hoare-atomic}; the proof of the call site uses \textsc{inv-open}.

\begin{figure}
\begin{subfigure}{.5\textwidth}
\begin{tabular}{c}
\begin{lstlisting}
let deref = λ l.
$\{ \pvs[\top][\emptyset] \exists w.\, l \mapsto w \ast (l \mapsto w \vsW[\emptyset][\top] \Phi(w)) \}$
  $\{ \exists w.\, l \mapsto w \ast (l \mapsto w \vsW[\emptyset][\top] \Phi(w)) \}$
  !l
  $\{v.\, v = w \ast l \mapsto w \ast (l \mapsto w \vsW[\emptyset][\top] \Phi(w) \}$
  $\{v.\, \pvs[\emptyset][\top] \Phi(v) \}$
$\{ \Phi(v) \}$
in
let x = ref 1 in
$ \{ R \ast \knowInv{\namesp}{I} \}~\textrm{where}~I \eqdef \exists n.\, x \mapsto n \ast (n = 1 \lor (n = 0 \ast R))$
fork (
  $\{ \knowInv{\namesp}{I} \}$
  while not CAS x 0 1 do () done;
  $\{ R \ast \knowInv{\namesp}{I} \}$
  1 / deref x; x := 0 (* $\cursivecomment{See main thread}$ *)
)
let v =
$\{ R \ast \knowInv{\namesp}{I} \}$
$\{ \pvs[\top][\emptyset] x \mapsto 1 \ast (x \mapsto 1 \vsW[\emptyset][\top] \Phi(1)) \}~\textrm{where}~\Phi(u) \eqdef R \ast u = 1$
  deref x
$\{ w.\, \Phi(w) \}$
in
$\{ v = 1 \ast R \ast \knowInv{\namesp}{I} \}$
1 / v; x := 0 (* $\cursivecomment{See Figure~\ref{fig:safety-proof-inlined}}$ *)
\end{lstlisting}
\end{tabular}
\caption{Proof outline for the program in Figure~\ref{fig:safety-seplogic-inlined} with the atomic $\deref$ replaced by $\texttt{deref}$.}%
\label{fig:safety-proof-deref}
\end{subfigure}%
\hfill
\begin{subfigure}{.4\textwidth}
\begin{mathpar}

  \axiomH{inv-alloc}{\later\prop \proves \pvs[\emptyset] \knowInv{\namesp}{\prop}}

  \inferH{inv-open}
    {\namesp \subseteq \mask}
    {\knowInv{\namesp}{\prop}\!\!\! \vs[\mask][\mask\setminus\namesp] \later\prop * (\later\prop \vsW[\mask\setminus\namesp][\mask] \TRUE)}

  \axiom{A \vsW[\mask_1][\mask_2] B \eqdef A \wand \pvs[\mask_1][\mask_2]}

  \axiom{A \vs[\mask_1][\mask_2] B \eqdef \square(A \vsW[\mask_1][\mask_2] B)}

  \axiom{A \vs[\mask] B \eqdef A \vs[\mask][\mask] B}

  \inferH{hoare-atomic}
  { \atomic(\expr) \\ \hoare{P}{e}{\pvs[\mask_2][\mask_1] Q}[\mask_2] }
  { \hoare{\pvs[\mask_1][\mask_2] P}{e}{Q}[\mask_1] }

  \axiom
    {\later (P \ast Q) \provesIff \later P \ast \later Q}

  \infer
    {\tau~\text{is inhabited}}
    {\later (\exists x : \tau.\, P) \provesIff \exists x : \tau.\, \later P}

  \axiom
    {\timeless{\laterCredit{n}}}

  \axiom
    {\later P \ast \laterCredit{1} \proves \pvs P}
\end{mathpar}
\caption{Selected proof rules. }%
\label{fig:inv-detailed-rules}
\end{subfigure}
\caption{Specifying fine-grained concurrent functions}
\end{figure}

\emph{Verifying \spinlock}
\Spinlock release can be specified analogously:\\
  \( \hoare{\pvs[\top][\emptyset] \lockvar \mapsto \mytrue \ast (\lockvar \mapsto \myfalse \vsW[\emptyset][\top] \Phi)}{\texttt{release}\spac\lockvar}{\Phi} \).
Specifying \spinlock acquire, however, is slightly trickier: we need a mechanism to allow the function to retry (busy-wait) while it is blocked (marked in green):\footnote{$\atomicupdate$ stands for ``atomic update''.}
\(\setlength\fboxsep{1pt} \hoare{\smash{\colorbox{greenbg}{$\atomicupdate$}}}{\texttt{acquire}\spac\lockvar}{\Phi} \)
where\\
\( \setlength\fboxsep{1pt} \colorbox{greenbg}{$\atomicupdate$} \eqdef \mathsf{gfp}(X.\;\pvs[\top][\emptyset] \exists b.\, \lockvar \mapsto b \ast \bigl( (\colorbox{greenbg}{$\lockvar \mapsto b \vsW[\emptyset][\top] X$}) \land ({b = \myfalse \wand \lockvar \mapsto \mytrue \vsW[\emptyset][\top] \Phi}) \bigr)) \)
and $\mathsf{gfp}(X.\,F(X))$ is the greatest fixpoint of $F$, i.e.~$\mathsf{gfp}(X.\,F(X)) \eqdef \bigvee P.\;(P \vdash F(P))$.

\subsection{Logically atomic triples}%

\begin{figure}
  \begin{subcaptionbox}
    {Encoding of logically atomic triples. Presentation adapted from~\citet{DBLP:journals/pacmpl/MulderK23}.
      \label{fig:lat-def}}
    {\input{./figures/lat/lat.tex}}
  \end{subcaptionbox}

  \begin{subcaptionbox}
    {Proof outline for client of concrete \spinlock.\label{fig:safety-proof-concrete-spinlock}}[.35\textwidth]
    {
\begin{lstlisting}
let x = ref 1 in
$\{ \knowInv{\namesp}{I} \ast R \}$
fork (
  $\{ \knowInv{\namesp}{I} \}$
  acquire x;
  $\{ \knowInv{\namesp}{I} \ast R \}$
  1/ !x; release x
)
$\{ \knowInv{\namesp}{I} \ast R \}$
1 / !x;
release x
$\{ \knowInv{\namesp}{I} \}$
\end{lstlisting}}
  \end{subcaptionbox}
  \begin{subcaptionbox}{Client invariant and atomic updates. Because the \spinlock implementation in Figure~\ref{fig:spinlock-classic-sassy} uses booleans we implicitly add the constraint $n \in \set{0,1}$.\label{fig:safety-proof-concrete-spinlock-defs}}{
  \def\gets{\coloneq}
  \(\begin{aligned}
    & \langle b.\, \lockvar \mapsto b \rangle~ \texttt{acquire}\spac\lockvar ~\langle \lockvar \mapsto \mytrue \ast b = \myfalse \rangle_{\emptyset} \\
    & \langle \lockvar \mapsto \mytrue \rangle~ \texttt{release}\spac\lockvar ~\langle \lockvar \mapsto \myfalse \rangle_{\emptyset} \\ 
    & R \eqdef{} \ownGhost{\gname}{\exinj()}, I \eqdef \exists n.\, x \mapsto n \ast (n = 1 \lor (n = 0 \ast R)) \\
    & \knowInv{\namesp}{I} \proves{} \langle n.\, x \mapsto n \mid x \mapsto 1 \ast n = 0 \vs R \rangle_{\emptyset} \\
    & \knowInv{\namesp}{I} \ast R \proves{} \langle x \mapsto 1 \mid x \mapsto 0 \vs \TRUE \rangle_{\emptyset} \\
    & \knowInv{\namesp}{I} \proves{} \hoare{R}{x \gets 0}{\TRUE}  \\
    & \knowInv{\namesp}{I} \proves{} \hoare{R}{\deref x}{v.\, v = 1 \ast R}
    \end{aligned}\)%
  }
  \end{subcaptionbox}
  \caption{\citet{DBLP:conf/popl/JungSSSTBD15}'s logically atomic triples.}%
  \label{fig:lat-and-use-case}
\end{figure}

The specification derived for \spinlock's \lstinline{acquire} is a special case of \citet{DBLP:conf/popl/JungSSSTBD15}'s \emph{logically atomic triples}.%
  \footnote{\citet{DBLP:journals/pacmpl/BirkedalDGJST21} have shown that logical atomicity implies linearizability~\cite{DBLP:journals/toplas/HerlihyW90}.}
Using their notation%
  \footnote{Inspired by \citet{DBLP:conf/ecoop/PintoDG14}'s atomic triples.}
  we specify \lstinline{acquire} in terms of its \emph{atomic pre- and postcondition}:
  \(\langle b.\, \lockvar \mapsto b \rangle~\texttt{acquire}\spac\lockvar~\langle \lockvar \mapsto \mytrue \ast b = \myfalse \rangle_{\emptyset} \).
Figure~\ref{fig:lat-and-use-case} shows the general encoding of logically atomic triples (\ref{fig:lat-def}) and their application to the program of Figure~\ref{fig:safety-seplogic-spinlock} (\ref{fig:safety-proof-concrete-spinlock}, \subref{fig:safety-proof-concrete-spinlock-defs}).

\subsection{Data abstraction}
This usage of the lock exposes the underlying boolean reference all but dictating the implementation as \spinlock.
We achieve data abstraction by letting the lock module define a lock predicate $L(\lockvar,b,\gname)$ of three arguments, whose definition is opaque to the client.
To lessen the proof burden for the client, in particular because $L$ may not be timeless, only $\later L(\ldots)$ is required in the atomic precondition.
\begin{align*}
  \hoare{\TRUE}{ &\texttt{create}\spac\TT }{\lockvar.\, \exists \gname.\, L(\lockvar,\myfalse,\gname)} \\
  \langle b.\, \later L(\lockvar,b,\gname) \rangle~ &\texttt{acquire}\spac\lockvar ~\langle L(\lockvar,\mytrue,\gname) \ast b = \myfalse \rangle_{\mask} \\
  \langle \later L(\lockvar,\mytrue,\gname) \rangle~ &\texttt{release}\spac\lockvar ~\langle L(\lockvar,\myfalse,\gname) \rangle_{\mask}
\end{align*}
Now, $\lockvar$ is an opaque reference and $b$ represents the \emph{abstract} lock state, $\myfalse$ for unlocked and $\mytrue$ for locked.
The third argument $\gname$ is new and can be used by the module to keep track of internal ghost state across function calls.
Because of its irrelevance to the client we usually omit it.
Finally, $\mask$ is the module mask reserved for operations (module internal invariants) on the lock.

\section{Total correctness logically atomic triples for unfair modules}%
\label{sec:tclat-unfair}

Equipped with this background on verifying partial correctness of uninstrumented programs we now turn towards instrumented ones.
In particular, we consider programs following the higher-order programming pattern of Section~\ref{sec:modular-programming-in-heaplanglt}.
Clients of blocking modules pass auxiliary code wrapped in auxiliary functions to blocking modules to abstract over the client-specific manipulation of auxiliary state.
Previously we reasoned informally about $\abeta$---required to produce fuel when the module is blocked on the client---and $\aalpha$, allowing the client to change auxiliary state the moment the blocking operation takes effect.
In this section we concretize these intuitions with precise contracts.

\subsection{\Spinlock-specific total correctness specification}

We want a modular yet expressive specification for the instrumented \spinlock of Figure~\ref{fig:sassy-spinlock}.
Recall that we use atomic blocks to bundle the execution of $\abeta$ and $\aalpha$ with real, atomic, expressions.
For example executing $\aalpha$ the moment the \spinlock is released \lstinline{/*⟨κ ();*/ lk := false/*⟩*/} (Figure~\ref{fig:sassy-spinlock}).
If the sequential expression $e$ is executed in a big-step---i.e.\ as part of an atomic block---$\bswpre{\expr}[\mask]{\Phi}$ describes the weakest precondition of $e$ such that $\Phi$ holds afterwards.

The logically atomic triples of Figure~\ref{fig:safety-proof-concrete-spinlock-defs} do not work for two reasons.
First, $\vscommit$ needs to say that it is safe to execute $\aalpha$, that $\aalpha$ terminates, and that after running $\aalpha$ the invariants can be closed again.
Therefore we adapt \vscommit as follows.
\begin{align*}
  b = \myfalse \wand{} &\lockvar \mapsto \mytrue \vsW[\emptyset][\top] \Phi & \text{partial correctness}\\
  b = \myfalse \wand \textlog{wp}^{\Downarrow}_{\emptyset}\spac {\aalpha} \spac \{&\lockvar \mapsto \mytrue \vsW[\emptyset][\top] \Phi \} & \text{total correctness}
\end{align*}

Second, $\vsabort$ needs to say that if the lock is already held, it is safe for the module to run $\abeta$, $\abeta$ terminates, and after running $\abeta$ the invariants can be closed again.
Separately, upon return, $\abeta$ must guarantee that a call permission to fuel the busy-waiting loop is available.
Therefore, we adapt $\vsabort$ as follows.
\begin{align*}
  & \lockvar \mapsto b \vsW[\emptyset][\top] \atomicupdate & \text{partial correctness} \\
  b = \mytrue \wand \textlog{wp}^{\Downarrow}_{\emptyset}\spac {\abeta}\spac \{\callp{\degree[0]} \ast{} (&\lockvar \mapsto b \vsW[\emptyset][\top] \tatomicupdate) \} & \text{total correctness}
\end{align*}
The important addition is the call permission in the postcondition.%
Note that \lstinline{acquire} consists of a simple busy-waiting \lstinline{CAS}-loop so that a call permission of minimal degree $\degree[0]$ per iteration suffices.
We use $\callp{\degree}[\theta]$ to represent ownership of a call permission of degree $\degree$ associated with thread $\theta$.
If $\theta$ is the current thread we omit it.

A \spinlock-specific total correctness specification thus looks as follows.
\begin{align*}
  & \langle b \eventually \myfalse.\, \lockvar \mapsto b\rangle~ \texttt{acquire}\spac\lockvar~\langle \lockvar \mapsto \mytrue \ast b = \myfalse \rangle_{\emptyset} \eqdef
  \forall \Phi, \abeta, \aalpha.\, \tatomicupdate \wand \wpre {\texttt{acquire}\spac\lockvar\spac\abeta\spac\aalpha} {\Phi}  \\&
  \text{where~} \tatomicupdate \eqdef \mathsf{gfp}(X.\, \pvs[\top][\emptyset] \exists b.\, \lockvar \mapsto b \ast
    \left( \begin{array}{l}
      (b = \mytrue \wand \bswpre{\abeta} [\emptyset] {\callp{\degree[0]} \ast (\lockvar \mapsto b \vsW[\emptyset][\top] X)}) \land {}\\
      (b = \myfalse \wand \bswpre {\aalpha} [\emptyset] {\lockvar \mapsto \mytrue \vsW[\emptyset][\top] \Phi})
    \end{array}\right))
\end{align*}
The eventually arrow%
  \footnote{This notation is inspired by \tadalive~\cite{DBLP:journals/toplas/DOsualdoSFG21} where it describes a constraint on the environment that always eventually the atomic precondition is satisfied.}
  $\eventually$ expresses the liveness requirement imposed on the client.
If unsatisfied, the module can legitimately expect the client to eventually put the module in a satisfying state.
This summarizes the core idea of our \emph{total correctness logically atomic triples}.

\subsection{Total correctness logically atomic triples}%
\label{sec:tclat-unfair-full}

\begin{figure}
  \input{./figures/lat/tclat-unfair.tex}
  \caption[]{Total correctness logically atomic triple for unfair, blocking, modules.
  }%
  \label{fig:tclat-unfair}
\end{figure}

Figure~\ref{fig:tclat-unfair} generalizes this specification approach to arbitrary \emph{unfair} blocking operations.
To start with, we generalize the eventually arrow notation to handle sequences of variables: $\vec{x} \eventually X$ (Figure~\ref{fig:tclat-unfair}, $\mycircled{\tusEa}$).
It engenders the preconditions $\vec{x} \not\in X$ $\mycircled{\tusEcond}$ and $\vec{x} \in X$ $\mycircled{\tusCcond}$ for $\abeta$ and $\aalpha$, respectively (cf.\ $b \eventually \myfalse$).

Operations may generally wish to block internally i.e.\ without involvement of the client.
In \heaplanglt, this requires that the levels of the signals they expect be below the levels of the thread's obligations.
The exclusive resource $\obligations{O}[\theta]$ $\mycircled{\tusPpr}$ asserts that thread $\theta$'s set of obligations is \emph{exactly} O.
Once again we omit $\theta$ if it matches the current thread.

\emph{Total correctness specification for unfair locks.}

We characterize unfair locks with the following specification.
\begin{align*}
\hoare{\TRUE}{&\texttt{create}\spac\TT}{\lockvar.\, L(\lockvar,\myfalse)} \\
{}^{\lev}_{\deg}\langle b \eventually \set{\myfalse}\!.\, \later L(\lockvar,b) \rangle~ &\texttt{acquire}\spac\lockvar ~\langle L(\lockvar,\mytrue) \ast b = \myfalse \rangle_{\mask} \\
{}_{\deg}\langle \later L(\lockvar,\mytrue) \rangle~ &\texttt{release}\spac\lockvar ~\langle L(\lockvar,\myfalse) \rangle_{\mask}^{\lev}
\end{align*}

Note that the triple for \lstinline{release} is slightly different than the one for \lstinline{acquire}.
For one, \lstinline{release}'s specification lacks the eventually arrow.
This means it is only passed $\aalpha$.
Moreover, $\lev$ is positioned to the right for \lstinline{release} and to the left for \lstinline{acquire}.
The distinction drawn here is whether the module has access to $\lev \prec O \ast \obligations{O}$ before (left) or after (right) running $\aalpha$.
Consequently, \lstinline{release}'s postcondition is $\Phi \ast \lev \prec O \ast \obligations{O}$ as opposed to just $\Phi$.

\section{Modular programming of fair modules}%
\label{sec:modular-programming-of-fair-modules}

\begin{figure}
    \begin{tabular}{c}
\begin{lstlisting}
let f = ref true in let lk = create () in
/*let sf = NewSignal 2 in let s = ref nil in*/
\end{lstlisting}
    \end{tabular}\\
    \begin{tabular}{c||c}
\begin{lstlisting}
/*let r = ref nil in*/
acquire lk /*
  (λ. if !r != !s then (
      NewExpectPerm !s ⊤C ⊥C;
      r := !s);
    Expect !s)*/;
  /*(λ. (s := NewSignal 0))*/;
f := false;
release lk /*(λ. SetSignal !s; SetSignal sf)*/
\end{lstlisting} &
\begin{lstlisting}
/*NewExpectPerm sf δ'' δ';
let r = ref nil in*/
let b = ref true in
while !b do
  acquire lk /*
    (λ. if !r != !s then (
        NewExpectPerm !s ⊤C ⊥C;
        r := !s);
      Expect !s)
    (λ. s := NewSignal 0)*/;
  b := !f;
  release lk
    /*(λ. SetSignal !s;
      if !f then (Expect sf;
        lower δ' into 2 times δ))*/
done
\end{lstlisting}
    \end{tabular}
  \caption{Instrumented version of \citet[\S~4]{DBLP:journals/toplas/DOsualdoSFG21}'s distinguishing client.
    If \lstinline{lk} is fair, the program is guaranteed to terminate under any fair schedule.
    Otherwise, the program may diverge.
    Notice that in this program the right thread can be bypassed at most once.}%
  \label{fig:distinguishing-client}
\end{figure}

The programming and specification patterns proposed in Sections~\ref{sec:modular-programming-in-heaplanglt} and~\ref{sec:tclat-unfair} are appropriate when the responsibility for ensuring that the blocking operation eventually finds the atomic precondition satisfied, is delegated entirely to the client.
In many cases, however, the distribution of responsibilities is more intricate.

Consider, for example, the erasure of the program in Figure~\ref{fig:distinguishing-client}.
In this program, there is no way for the client to guarantee, on their own, that the left thread will eventually find itself able to acquire the lock.
Indeed, this is only the case if the left thread's \lstinline{acquire} cannot be \emph{starved} by the lock implementation.
We say that an \lstinline{acquire} operation is \emph{bypassed} if during its execution the lock is released and subsequently acquired by another \lstinline{acquire}.
A lock is \emph{starvation-free} if each \lstinline{acquire} can only be bypassed finitely many times.
We call such locks \emph{fair}.
Concretely, if the lock is unfair, e.g.\ the \spinlock in Figure~\ref{fig:spinlock-classic-sassy}, this program may diverge.
But if it is a fair lock, e.g.\ the \ticketlock in Section~\ref{sec:ticketlock}, it terminates under any fair scheduler.

The division of responsibilities between module and client in case of fair locks, then, is as follows.
It is the client's responsibility to ensure that all bypassing acquire operations are eventually paired with a corresponding release.
The module, in turn, must ensure a finite number of bypasses before a given acquire operation succeeds.

In terms of our instrumentation, this distinction manifests as the program in Figure~\ref{fig:distinguishing-client} being safe or not.
With an unfair lock, the left thread's $\abeta$ has to guarantee possibly infinitely many \lstinline{release}s---one for each bypass.
To do so, it would need to create infinitely many expect permissions which is impossible (Section~\ref{sec:overview}).
Conversely, the auxiliary code is safe if a fair lock module ensures that for every bypass---of which there can be only finitely many---there exists a compensatory call permission.

\subsection{\Ticketlock}%
\label{sec:ticketlock}

\begin{figure}
\def\ownfld{\mathit{owner}} %
\def\nextfld{\mathit{next}} %
\def\heldfld{\mathit{held}} %
\def\mapnname{\mathsf{M}} %
\def\levelsmod{\Lev_{\mathsf{TL}}}
\def\gheld{\gname_{\mathsf{held}}}
\def\maskmod{\mask_{\mathsf{TL}}}
\begin{tabular}{cc}
\begin{lstlisting}
let acquire lk /*η κ*/ =
  let t =
  /*⟨if !(lk.next) = !(lk.owner)
   then (* $\cursivecomment[myghostcode]{\textbf{Acq-LP-1}}$ *) κ () 
   else (FAIRNESS;
     lk.κs[!(lk.next)] := κ);*/
   FAA lk.next 1/*⟩*/
  in
  while
  /*⟨if !(lk.owner) != t then η () else ();*/
   !(lk.owner) != t/*⟩*/
  do () done
\end{lstlisting} &
\begin{lstlisting}
let release lk /*κ*/ =
  /*⟨ (* $\cursivecomment[myghostcode]{\textbf{Rel-LP}}$ *) κ ();*/
   FAA lk.owner 1/*;
   if !(lk.owner) < !(lk.next)
   then (* $\cursivecomment[myghostcode]{\textbf{Acq-LP-2}}$ *)
     lk.κs[!(lk.owner)] () ⟩*/
\end{lstlisting}
\end{tabular}
\caption{Instrumented \ticketlock implementation. The ticket \lstinline{t} is the old value returned when \lstinline{lk.next} is bumped. \textsc{fairness} is a placeholder for the auxiliary code to satisfy the fair lock contract with respect to $\abeta$.}
\label{fig:ticketlock-impl-aux}
\end{figure}

Figure~\ref{fig:ticketlock-impl-aux} shows a \ticketlock~\cite{DBLP:journals/tocs/Mellor-CrummeyS91} implementation instrumented in the style of Section~\ref{sec:modular-programming-in-heaplanglt}.
The \ticketlock is a simple fair lock consisting of two fields: \lstinline{owner} ticket and \lstinline{next} available ticket, both initialized to zero.
During \lstinline{acquire}, the \lstinline{next} field is atomically incremented with the fetch-and-add (\lstinline{FAA}) instruction.
\lstinline{FAA} returns the old value which becomes the ticket.
Thereafter, the owner field is read in a busy-waiting loop until the read value matches the ticket.
Releasing the lock simply increments the owner counter by one.%
  \footnote{For simplicity we use \lstinline{FAA} for this operation too.}
 
Before revealing the auxiliary code that \textcolor{myghostcode}{\textsc{fairness}} stands-in for, we discuss the placement of $\abeta$ and $\aalpha$ as application of the programming pattern of Section~\ref{sec:modular-programming-in-heaplanglt}.
First, \lstinline{acquire} and \lstinline{release}'s $\aalpha$s must be executed at their respective linearization points.
Let $t$ be a ticket and $o$ and $n$ be the contents of \lstinline{owner} and \lstinline{next}, respectively.
The \ticketlock is held by the thread owning ticket $t$ iff $t = o$.
This implies two linearization points for \ticketlock \lstinline{acquire}:
\begin{inparaenum}[(1)]
  \item An internal one if \lstinline{FAA lk.next 1} returns the ticket $o$.
  Because \lstinline{FAA} returns the old value, $t = n = o$, while \lstinline{next}'s new value is $n + 1$.
  Therefore the lock linearizes at this moment, transitioning from not held to held.
  This situation is detected in the auxiliary code, executing $\aalpha$ accordingly (Figure~\ref{fig:ticketlock-impl-aux}, $\cursivecomment[myghostcode]{\textbf{Acq-LP-1}}$).
  \item An external one if the returned ticket $t < o$. In this case, \lstinline{acquire} linearizes at the \lstinline{FAA} of \lstinline{release} when \lstinline{owner} equals $t - 1$ ($\cursivecomment[myghostcode]{\textbf{Acq-LP-2}}$).
\end{inparaenum}
Regardless of whether \lstinline{release} linearizes an \lstinline{acquire} or not, its own linearization point happens with the increment of \lstinline{owner}.

Second, $\abeta$ is called when \lstinline{acquire} establishes that it is blocked on the client because its ticket $t$ does not equal the contents of \lstinline{owner}.
Notice that if \lstinline{acquire} linearizes at $\cursivecomment[myghostcode]{\textbf{Acq-LP-1}}$, $\abeta$ will not be called.
Otherwise, \lstinline{acquire} will see itself bypassed by all preceding ticket holders.

Indeed, the ticket regime enforces a first-in-first-out (FIFO) direct handoff succession policy~\cite{DBLP:conf/eurosys/Dice17,DBLP:journals/tocs/GuerraouiGLQT19}.
The auxiliary code placeholder thus compensates exactly that many bypasses.
\lstinline{/*(* $\cursivecomment[myghostcode]{FAIRNESS}\eqdef$ *) lower δ into !(lk.next) - !(lk.owner) times ⊤C*/}
In the \appendixorsupplement we describe how to pick $\degree$ and $\top_{\mathsf{eta}}$ modularly.

\section{Total correctness logically atomic triples for fair modules}%
\label{sec:tclat-fair}

To match the extended programming pattern of Section~\ref{sec:modular-programming-of-fair-modules} we adapt the notion of total correctness logically atomic triples to fair modules, here locks.
We say that every time the client fulfills its duty by calling \lstinline{release} a new \emph{round} starts:
\( \setlength\fboxsep{1pt} \langle L(\lockvar,\mytrue,\colorbox{greenbg}{$r$}) \rangle~ \texttt{release}\spac\lockvar ~\langle L(\lockvar,\mytrue,\colorbox{greenbg}{$r + 1$}) \rangle_{\mask} \).
Otherwise the specification for a fair \lstinline{release} is identical to an unfair one.

At the notation-level the impact of fairness on \lstinline{acquire} is similarly minimal.
\( \setlength\fboxsep{1pt} \langle (b,\colorbox{greenbg}{$r$}) \eventually[\colorbox{greenbg}{$r$}] \set{\myfalse} \times \nat.\, L(\lockvar,b,\colorbox{greenbg}{$r$}) \rangle~\texttt{TL\_acquire}\spac\lockvar~\langle L(\lockvar,\mytrue,r) \ast b = \myfalse \rangle_{\mask} \).
The addition of the rounds subscript to the eventually arrow $\eventually[r]$ indicates which part of the atomic precondition evolves to indicate new rounds.
All but one competing \lstinline{acquire} operation will see themselves bypassed for each new round and must thus be compensated with a call permission.
To show the difference with respect to the unfair encoding, we first recall the definition of our total correctness logically atomic triple for unfair locks e.g.\ \spinlock.%
  \footnote{We defer the treatment of obligations (i.e.\ the $\Lev_M$ annotation) because both \spinlock and \ticketlock's blocking behavior is entirely justified by the client.}
{
\def\modid{\mathsf{SL}}
\begin{align*}
  & \langle b \eventually \myfalse.\, L(\lockvar, b) \rangle~ \texttt{SL\_acquire}\spac\lockvar~\langle L(\lockvar,\mytrue) \ast b = \myfalse \rangle_{\mask} \eqdef \\&\qquad
  \forall \Phi, \abeta, \aalpha.\, (\callp{\top_\modid} \ast \tatomicupdate) \wand \wpre {\texttt{SL\_acquire}\spac\lockvar} {\Phi} \\&
  \text{where~} \tatomicupdate \eqdef \mathsf{gfp}(Z.\, \pvs[\top\setminus\mask][\emptyset] \exists b.\, L(\lockvar,b) \ast
    \left(\begin{array}{@{} l @{}}
      (b = \mytrue \wand \bswpre{\abeta}[\emptyset]{\callp{\top_\mathsf{SL}} \ast (L(\lockvar,b) \vsW[\emptyset][\top\setminus\mask] Z)}) \land {}\\
      (b = \myfalse \wand \bswpre {\aalpha} [\emptyset] {L(\lockvar,\mytrue) \vsW[\emptyset][\top\setminus\mask] \Phi})
    \end{array}\right))
\end{align*}
}
We use the round resource $R$ to encode that $\abeta$ requires one-time compensation per round.
{
\def\modid{\mathsf{TL}}
\begin{align*} &\setlength\fboxsep{1pt}
  \langle (b,\colorbox{greenbg}{$r$}) \eventually[\colorbox{greenbg}{$r$}] \set{\myfalse} \times \nat.\, L(\lockvar,b,\colorbox{greenbg}{$r$}) \rangle~\texttt{TL\_acquire}\spac\lockvar~\langle L(\lockvar,\mytrue,r) \ast b = \myfalse \rangle_{\mask} \eqdef{} \\ & \setlength\fboxsep{1pt}\qquad
  \forall \Phi,\abeta,\aalpha,\colorbox{greenbg}{$R$}.\, (\callp{\top_\modid} \ast \colorbox{greenbg}{$R(\_)$} \ast \tatomicupdate) \wand \wpre{\texttt{TL\_acquire}\spac\lockvar}{\Phi} \\&\setlength\fboxsep{1pt}
  \tatomicupdate \eqdef \mathsf{gfp}(Z.\, \pvs[\top\setminus\mask][\emptyset] \exists b,\colorbox{greenbg}{$r$}.\, L(\lockvar,b,\colorbox{greenbg}{$r$}) \ast{} \\ &\setlength\fboxsep{1pt}
  \left(\begin{array}{@{} l @{}}
  \bigl((b = \mytrue \ast (\colorbox{greenbg}{$\exists r_0.\, R(r_0) \ast (r_0 = r \lor \callp{\top_C})$}) \wand \bswpre{\abeta}[\emptyset]{\callp{\colorbox{greenbg}{$\bot_C$}} \ast \colorbox{greenbg}{$R(r)$} \ast (L(\lockvar,b,\colorbox{greenbg}{$r$}) \vsW[\emptyset][\top\setminus\mask] Z)}\bigr) \\
  {}\land \bigl((b = \myfalse \ast (\colorbox{greenbg}{$\exists r_0.\, R(r_0)$}) \wand \bswpre{}[\emptyset]{L(\lockvar,\mytrue,\colorbox{greenbg}{$r$}) \vsW[\emptyset][\top\setminus\mask] \Phi} \bigr)
  \end{array}\right))
\end{align*}
}
Notice that $R$ has to be treated ``linearly'' by the module.
Initially passed to the module at an unknown round, every call to $\abeta$ takes and returns $R$ in its pre- and postcondition until $\aalpha$ consumes it at the linearization point.
Calling $\abeta$ at round $r$ returns $R(r)$.
Future calls to $\abeta$ are free if the module can prove to $\abeta$ that it has already been called for the current round: $\exists r_0.\, R(r_0) \ast r_0 = r$ (recall that $r$ always expresses the current round).
Otherwise $\abeta$ requires a compensatory call permission.

\emph{Internal waiting.}
\begin{figure}
  \input{./figures/lat/tclat-fair.tex}
  \caption{Total correctness logically atomic triple for fair blocking modules.}%
  \label{fig:tclat-fair}
\end{figure}
Figure~\ref{fig:tclat-fair} generalizes this encoding to our encoding of total correctness logically atomic triples for \emph{fair} modules.
To support internal blocking we again require the caller to pass $\obligations{O}$ with the lower bound determined by a set of module degrees $\Lev_M$ ($\mycircled{\tfsLm}$,$\mycircled{\tfsOPpre}$,$\mycircled{\tfsOpreK}$).
In Section~\ref{sec:tclat-unfair-full} all changes to the thread's obligations had to be opaque to the client: $\abeta$ and $\aalpha$ require $\obligations{O}$ in their precondition.
For fair modules%
  \footnote{We could also view unfair modules as fair modules with a constant round.}
  we relax this restriction for $\abeta$ and allow any set of obligations $O'$ with $\Lev_M \preceq O'$ ($\mycircled{\tfsOpre}$, $\mycircled{\tfsOpost}$).
This permits calling $\abeta$ while holding internal obligations whose level is above or in $\Lev_M$.
Section~\ref{sec:cohortlock} discusses how the proof for a fair \cohortlock uses this affordance.

\emph{Total correctness specification for fair locks.}
\begin{align*}
& \hoare{\TRUE}{\texttt{create}\spac\TT}{\lockvar.\, L(\lockvar,\_,\myfalse)} \\
& {}^{\Lev_{\mathsf{FL}}}\langle (b,r) \eventually[r] \set{\myfalse}\times\nat.\, \later L(\lockvar,r,b) \rangle~ \texttt{acquire}\spac\lockvar ~\langle L(\lockvar,r,\mytrue) \ast b = \myfalse \rangle_{\mask} \\
& \langle \later L(\lockvar,r,\mytrue) \rangle~ \texttt{release}\spac\lockvar ~\langle L(\lockvar,r+1,\myfalse)\rangle_{\mask}^{\Lev_{\mathsf{FL}}}
\end{align*}

\section{Tool support}%
\label{sec:tool-support}

In this paper we used \iris\footnote{On paper only, not in \coq; a \coq mechanization of our metatheory is future work.} to verify safety of instrumented programs.
However, any other approach for verifying safety of higher-order fine-grained concurrent programs can be used as well, such as that implemented by the VeriFast tool for separation logic-based modular symbolic execution of concurrent C or Java programs \cite{DBLP:journals/corr/Vogels0P15}.
We used the existing encoding of Ghost Signals into VeriFast by \citet{DBLP:conf/cav/ReinhardJ20} to modularly verify termination of Java versions of the examples of this paper\footnote{See the \supplementorartifact.}, including the \cohortlock described in the next section.
While non-foundational---bugs in the axiomatization of Ghost Signals or in VeriFast itself could lead to unsoundness---these machine-checked example proofs dramatically increase our confidence that our approach can in fact be applied to modularly verify termination of highly non-trivial concurrent algorithms.\footnote{A mechanization of the example proofs in \iris in \coq would give even much greater confidence; this is future work.}

\section{\Cohortlock}%
\label{sec:cohortlock}

\begin{figure}
  \begin{subfigure}{.4\textwidth}
    \begin{tabular}{c}
\begin{lstlisting}
let acquire lk =
  let LL = lk.LLs[CC] in
  ticketlock_acquire LL;
  if lk.passing[CC] then
    (* $\color{black}\mathit{Handoff}$ *)
    lk.passing[CC] := false
  else
    ticketlock_acquire TL
\end{lstlisting}
    \end{tabular}
  \end{subfigure}%
  \hfill
  \begin{subfigure}{.5\textwidth}
    \begin{tabular}{c}
\begin{lstlisting}
let release lk =
  let LL = lk.LLs[CC] in
  if ticketlock_alone LL
    || lk.passCount[CC] >= MAX then
    lk.passCount[CC] := 0;
    ticketlock_release TL;
    ticketlock_release LL
  else
    lk.passCount[CC]++;
    lk.passing[CC] := true;
    ticketlock_release LL
\end{lstlisting}
    \end{tabular}
  \end{subfigure}
  \hfill\null
  \caption[\Cohortlock]{\Cohortlock. \lstinline{CC} means \emph{current cohort}.
    If there are waiting threads from the same cohort, the \cohortlock will pass ownership directly to the next thread of the cohort by releasing the local lock \lstinline{LLs[CC]} but not the top-level lock \lstinline{TL}.
    To avoid starvation, such handoffs are performed a maximum of \lstinline{MAX} times.
    The \cohortlock is fair if its underlying locks are fair.
  }%
  \label{fig:cohortlock}
\end{figure}

Lock cohorting~\cite{DBLP:conf/ppopp/DiceMS12,DBLP:journals/topc/DiceMS15} is a design pattern to construct non-uniform memory access (\numa) aware locks from \numa unaware locks.
On a \numa architecture, threads can be grouped into cohorts belonging to the same \emph{node}.
While all threads share a logical address space, access to memory last accessed by a thread from the same cohort is much faster than otherwise.
Therefore, if there are threads of multiple cohorts competing for a lock, it is preferable to pass the lock to a thread of the same cohort to improve throughput.

Figure~\ref{fig:cohortlock} shows an implementation in \heaplang.
Specifically, \cohortlock{}s employ one top-level lock shared by all cohorts, plus a local lock per cohort.
To own a \cohortlock, both the local lock and the top-level lock must be owned.
In addition to \lstinline{acquire} and \lstinline{release}, the local lock must support an \lstinline{alone}  operation that only returns false if there is a thread waiting to acquire the lock.%
  \footnote{\lstinline{alone} may return false positives but never false negatives.}
For the \ticketlock this is easy: \lstinline{let alone lk = !(lk.owner) + 1 = !(lk.next)}.
Thus, to prefer passing ownership within the same cohort, the \cohortlock's \lstinline{release} queries the underlying local lock whether the releasing thread is alone.
If so both the top-level and local lock are released (in that order).
Otherwise, the \cohortlock \lstinline{release} can choose to only release the local lock while passing ownership of the top-level lock to the next thread  in the cohort (cf.\ Figure~\ref{fig:motivating-example}).
To prevent starvation of other cohorts, it only does so until a threshold is reached.
Moreover, the \cohortlock is only fair if the underlying locks are fair.

We proved that this lock implementation satisfies total correctness per our fair lock specification, assuming only that the \ticketlock{}s used in turn satisfy our fair lock specification.
We (non-foundationally) machine-checked our proof using \verifast; see the \supplementorartifact.

A core design decision of the proof is that \lstinline|acquire| creates a
signal (which we call the \emph{acquire signal}) during the
\lstinline|ticketlock_acquire LL| call's linearization point (LP), but only if
\lstinline|lk.passing[CC]| is false, and discharges it during
\lstinline|ticketlock_acquire TL|'s LP. Similarly, \lstinline|release| creates
a signal (the \emph{release signal}) during the
\lstinline|ticketlock_release TL| call's LP (again, if not passing) and discharges it during the \lstinline|ticketlock_release LL| call's LP.
If the $\abeta$ passed into \lstinline|ticketlock_acquire LL| is executed, some other thread in the cohort holds \lstinline|LL|, and either 1) has not yet acquired \lstinline|TL|, 2) holds \lstinline|TL|, or 3) has released \lstinline|TL|.
Correspondingly, this $\abeta$ expects either this thread's acquire signal, calls the client's $\abeta$, or expects the thread's release signal.
If the $\abeta$ passed into \lstinline|ticketlock_acquire TL| is executed, some other cohort holds \lstinline|TL|. At most \lstinline|MAX| number of successive threads of that cohort own the lock, so at most \lstinline|MAX| rounds of the client's $\abeta$ have to be called.

\section{Related work}%
\label{sec:related-work}

We propose the first approach for expressive modular termination verification of busy-waiting programs in a higher-order separation logic like Iris or that of VeriFast. This means that it benefits from the existing metatheory development and tool support for those logics, and furthermore that it is compatible with the other features and many of the other applications and extensions of those logics proposed in the literature, and is more accessible to the communities already familiar with those logics.

\emph{\tadalive.}
\citet{DBLP:journals/toplas/DOsualdoSFG21}'s \tadalive is a bespoke first-order separation logic to compositionally reason about termination of fine-grained concurrent programs.
We believe the clear distinction between module and client responsibilities for termination advocated for---and demonstrated by---\tadalive set a new standard for modular termination verification.

In contrast to \sassy, however, it does not support unstructured concurrency or external linearization points.
Another notable difference is the interpretation of ``$\eventually$'': in \tadalive it expresses that the client must always eventually fulfill the atomic precondition.
Since this is insufficient to prove termination of unfair locks their unfair lock specification requires an \emph{impedance budget}, an ordinal that bounds the number of acquire invocations.
By contrast, our $\eventually$ requires the client to justify progress in terms of call permissions whenever the atomic precondition is not in the module-desired state. As a result, in contrast to \tadalive, \sassy supports program $\mathbb{C}_2$ \cite[p. 79]{DBLP:journals/toplas/DOsualdoSFG21}.\footnote{See our proof in the \supplementorartifact.}

A more serious limitation of \tadalive is that, like our approach, it associates a level (called a \emph{layer} in \tadalive) with each lock and requires that only environment obligations at levels below the lock's level are used to justify blocking. Unlike our approach, however, it also does not support asynchronously burdening another thread with obligations. We believe that, together, these limitations mean it does not support a variant of the client program of Figure~\ref{fig:motivating-example} where the middle thread busy-waits for the left thread to clear the flag, whereas our approach does.\footnote{For more details and our proof of this program, see the \supplementorartifact.}

\emph{\iris.}
 Although \iris~\cite{DBLP:conf/popl/JungSSSTBD15,DBLP:conf/icfp/0002KBD16,DBLP:journals/jfp/JungKJBBD18,DBLP:journals/pacmpl/SpiesGTJKBD22} is fundamentally limited to safety properties due to its use of step-indexing, many approaches have been developed to work around this limitation.
The \iris developments discussed below have been foundationally mechanized in \iris' \coq-framework.
Transfinite \iris~\cite{DBLP:conf/pldi/SpiesGGTKDB21} uses ordinals as step-index instead of natural numbers.
This solves the biggest issue of a naive fuel-based approach in \iris at the logic level, obviating the need for auxiliary code.
\citet{DBLP:conf/pldi/SpiesGGTKDB21} prove program termination via refinement of an ordinal in addition to termination-preserving refinement of another program.
Transfinite \iris does not, however, address blocking concurrent programs under fair scheduling.
\citet{DBLP:conf/esop/TassarottiJ017} proved termination-preserving refinement for concurrent programs, but with the rather severe limitation that stuttering in the refinement must be bounded by a fixed natural number.
Trillium~\cite{popl24-trillium} considers a more general form of refinement between traces of a program and a model.
Fairis, an instantiation of Trillium for reasoning about concurrent programs under fair scheduling, allows them to reason about liveness properties, including termination.
Their approach requires a terminating model, which is then proven to be refined by the program.
Modular composition of models is mentioned as future work.

\section{Conclusion and future work}%
\label{sec:conclusion}%
\label{sec:future-work}

We propose the first approach for expressive modular verification of termination of busy-waiting programs in a higher-order separation logic like \iris. We demonstrate its expressiveness and amenability to tool support by reporting on a (non-foundational) machine-checked proof that a \cohortlock---built on top of other locks and using lock handoff internally (where a lock is acquired by one thread and released by another)---satisfies our fair lock specification, assuming only that the underlying locks it is built on also satisfy our fair lock specification. It is a two-step approach. First, the program is instrumented modularly with auxiliary constructs from a fuel-based terminating programming language that has the property that if an instrumented program is safe, the original program is also safe and terminates under fair scheduling. Secondly, any logic for verification of safety of higher-order fine-grained concurrent programs is used to modularly verify safety of the instrumented program. We propose auxiliary code patterns for modularly instrumenting unfair and fair blocking modules, and corresponding specification patterns for modularly verifying the instrumented programs' safety.

We identify three key directions for future work.
First, in this paper we presented how \iris can be used to reason about \heaplanglt programs.
The implementation and proof of the \ticketlock leveraged \heaplang's higher-order nature and \iris' ability to reason about helping.
Another key strength of \iris with respect to the verification of concurrent modules is its support for erasable prophecy variables to reason about future-dependent linearization points.
\heaplanglt inherits support for prophecy variables but cannot read them in erasable code.
To enable verifying total correctness of concurrent modules with future-dependent linearization points, we plan to modify prophecy variables, allowing auxiliary code to read them.
We suspect that this might additionally enable attaching prophecy resolutions to logically atomic operations, not just atomic operations, a current limitation~\cite{DBLP:journals/pacmpl/JungLPRTDJ20}.
Second, a natural next step would then be to integrate the presented work and this extension into a \coq mechanization.
This would give foundational backing to our approach.

\begin{acks}
This research is partially funded by the Research Fund KU Leuven, and by the Cybersecurity Research Program Flanders.
\end{acks}

\bibliography{ref}
\bibliographystyle{ACM-Reference-Format}

\clearpage

\appendixpage
\appendix
\section{Modular picking of degrees and levels}

\subsection{Unfair modules}

All modules of a \heaplanglt program are generally parameterized by some degrees among which some order must hold. If we assume the dependency relation on modules forms a tree, we can pick these degrees modularly as follows. Each module $M$ defines a \emph{local domain of degrees} $\Deg_M \subseteq \mathit{Val}$, equipped with a well-founded order ${<}_{\Deg_M}$. In defining its local domain of degrees, a module may refer to those of its dependencies. It also defines any \emph{named degrees} relevant to clients as elements of that local domain. For a client to obtain a module's code (which includes the code of its dependencies), it must supply a \emph{degree embedding} $\epsilon_M$, an order-preserving function that maps the module's local domain of degrees to the global universe of degrees $\Degdom$.

\emph{Instantiating the degrees.}
Recall that the example client program (Figure~\ref{fig:motivating-example-folded} in the main text) is parameterized by the degree $\delta_3$ of the call permission expected by the lock module's $\mathsf{acquire}$ implementation, both when it is invoked and after it invokes its $\abeta$ argument, as well as by degrees $\delta_1$ and $\delta_2$ such that $\delta_3 < \delta_2 < \delta_1$.

Let's apply this approach to our example program (Figure~\ref{fig:motivating-example-folded} in the main text).
The spinlock module simply defines $\Deg_\mathsf{LK} = \{0\}$ with $<_{\Deg_\mathsf{LK}} = \emptyset$. It defines named degree $\top_\mathsf{LK} = 0$. The client program defines its $\Deg_\mathsf{C} = \{0\} \times \Deg_\mathsf{LK} \cup \{(1, 0), (2, 0)\}$, ordered lexicographically. It also defines named degree $\top_\mathsf{C}$, the degree of the call permission it needs to run, as equal to $(2, 0)$. Given a degree embedding $\epsilon_\mathsf{C}$, the client module builds its code by first obtaining the code for the lock module by passing embedding $\epsilon_\mathsf{LK} = \lambda \delta.\, \epsilon_\mathsf{C}((0, \delta))$ and then instantiating its own code with $\delta_1 = \epsilon_\mathsf{C}((2, 0)), \delta_2 = \epsilon_\mathsf{C}((1, 0)), \delta_3 = \epsilon_\mathsf{C}((0, \top_\mathsf{LK}))$.

We can now use the client program's local order ${<}_{\Deg_\mathsf{C}}$ as the order on $\Degdom$ and instantiate the client program with $\epsilon_\mathsf{C} = \lambda \delta.\,\delta$ to obtain fully instantiated \heaplanglt code that we can run, starting with a call permission of degree $\top_\mathsf{C}$.

\emph{Instantiating the levels.}
We can apply the exact same approach to modularly determine a program's levels: $\Lev_\mathsf{C} = \{0\} \times \Lev_\mathsf{LK} \cup \{(0, 0), (1, 0)\}$, ordered lexicographically, with $\lev_1 = \epsilon_\mathsf{C}((0, 0))$ and $\top_\mathsf{C} = (1, 0)$. The client program requires that the caller's obligations are above (or equal to) $\top_\mathsf{C}$.

\subsection{Fair modules}

\emph{Instantiating the degrees.}
In Section~\ref{sec:modular-programming-of-fair-modules} of the main text we discuss modular programming of fair modules.
Since in the case of fair modules, call permissions flow from the client to the module, then back to the client's $\eta$, and then finally back to the module, we extend the approach for modularly determining a program's degrees to support this. A fair blocking module's local domain of degrees is no longer fixed; it now depends on the client's \emph{eta domain of degrees} $\Deg_M^\mathsf{eta}$, a subset of $\mathit{Val}$ equipped with a well-founded order as well as a distinguished top element and bottom element. Given a client's eta domain, the module provides a local domain as well as an embedding from the eta domain into the local domain. The client then provides an embedding of the local domain into the global universe to obtain the module's code, as before.

Applying this approach to the \ticketlock of Section~\ref{sec:ticketlock} in the main text:
\begin{itemize}
  \item The \ticketlock module is paramterized by the client's eta domain of degrees $\Deg_C^{\mathsf{eta}}$ with named degrees $\top_C^{\mathsf{eta}}$ and $\bot_C^{\mathsf{eta}}$
  \item The module defines its module-local degrees $\Deg_{M} \eqdef \set{0} \times \Deg_C^{\mathsf{eta}} \cup \set{(1,0)}$, ordered lexicographically.
    The named degree $\top_M = \top_{\mathsf{TL}} ({}= \degree~\text{in Section~\ref{sec:ticketlock} of the main text}) = (1,0)$.
  \item Additionally, the module defines the embedding $\epsilon_{M} \eqdef \lambda \degree.\, (0,\degree)$.
  \item Finally, the client provides the embedding $\epsilon_{C}$ into the universe of degrees.
\end{itemize}
Thus $\degree$ and $\top_{\mathsf{eta}}$ in
\lstinline{/*lower δ into !(lk.next) - !(lk.owner) times ⊤C*/}
refer to $\epsilon_{C}(\top_{\mathsf{TL}})$ and $\epsilon_{C}(\epsilon_{M}(\top_{C}^{\mathsf{eta}}))$, respectively.

\section{\heaplanglt: a terminating language}

\subsection{Absence of infinite fair executions}

We prove Theorem~\ref{thm:no-infinite-fair-executions} of the main text.
Due to the setup's similarity, the proof is similar to that of \citet{DBLP:conf/cav/ReinhardJ20,Reinhard2021GhostSignalsTR}.

\begin{theorem}[Absence of infinite fair executions]\label{thm:no-infinite-fair-executions-appendix}
  A \heaplanglt program of the form $e; \Finish$ does not have infinite fair executions.
\end{theorem}
\begin{proof}
By contradiction.
Thus, throughout this section---including the auxiliary definitions and lemma below---we assume an infinite fair execution of a \heaplanglt program $e; \Finish$, i.e.\ an infinite sequence of machine configurations related by machine steps where, if in a configuration at index $i$ a thread $\theta$ has not finished, then at some index $j \ge i$ that thread takes a step.
Additionally, let $\Sinf$ be the set of signals \Expected infinitely often in this execution.

\begin{definition}[Program order graph]
  We construct the infinite \emph{program order graph} (in fact, a tree) as follows.
  Each step in the infinite sequence, identified by its position $i$, corresponds to a node in the tree.
  Each node is connected with a directed edge to the next step of the same thread, and if the step is a fork, to the first step of the child.
\end{definition}

\begin{definition}[Path fuel]\label{def:path-fuel}
  Consider an infinite path in the program order graph starting at some index $i$ that never $\Expect$s a signal in $\Sinf$.
  We define the path fuel for this path at each machine step $j \geq i$ as the multiset of degrees obtained by combining the path's multiset of call permissions at configuration $j$ with $n - j + 1$ copies of degree $\delta$ for each of the path's expect permissions at configuration $j$ for signal $s$ at degree $\delta$ where signal $s$ is last \Expected on the path at step $n$:
  \begin{align*}
    \mathsf{PF}(j) \eqdef{} & \istate.\textsc{CallPerms}(\theta) \uplus{} \\
    & \biguplus_{ (s,\delta) \in \istate.\stateExpectPerms(\theta) }
        \left. \begin{cases}
            [(n - j + 1) \cdot \degree] & \text{$s$ last \Expected at step $n$} \\
            \emptyset & \text{otherwise, i.e.\ $s \in \Sinf$}
        \end{cases} \right. \\
    \text{where} \\
    \theta~\text{is}&~\text{the thread id of the first step on the path whose index is not less than $j$.} \\
  \istate~\text{is}&~\text{the machine state before step $j$.}
  \end{align*}
  Notice that $n - j + 1$ constitutes an upper bound on the number of times the path will \Expect this expect permission. Notice that if the path owns any expect permissions for signals in $\Sinf$, they are ignored, since, by our assumption, the path will not use them anyway.

  Like other multisets, we order path fuel values by the well-founded Dershowitz-Manna order~\cite{DBLP:journals/cacm/DershowitzM79,DBLP:journals/aaecc/Coupet-GrimalD06}.
\end{definition}

\begin{lemma}[Finite paths]\label{lem:finite-paths}
  A path that never $\Expect$s a signal in $\Sinf$ is finite.
\end{lemma}
\begin{proof}
  Suppose otherwise.
  Let $F$ be the set of nodes on the path that are function calls.
  Since \heaplanglt's only source of divergence are function calls, $F$ must be infinite.
  All steps in $F$ decrease the path fuel due to function calls consuming a call permission.
  All steps---on or off the path---also decrease the path fuel or leave it unchanged.
  Thus, the infinite path implies an infinite descending chain.
  A contradiction.
\end{proof}

We now continue the proof of Theorem~\ref{thm:no-infinite-fair-executions} by case splitting on whether the execution \Expects any signals infinitely often or not.\\
\textbf{Case 1}: $\Sinf = \emptyset$.
By Lemma~\ref{lem:finite-paths}, all paths of the program order graph must be finite.
However, applying K\H{o}nig's lemma we conclude that the infinite program order graph of the infinite fair execution must have at least one infinitely long path.
A contradiction.\\
\textbf{Case 2}: $\Sinf \neq \emptyset$.
Let $\smin \in \Sinf$ be a signal of minimal level in $\Sinf$.
Such a signal must exist since the universe of levels is well-founded.
Let $p$ be the path that carries the obligation to set $\smin$ i.e.\ the path constituted by the set of steps by a thread while holding the obligation for $\smin$.
This path must be infinite: because signal $\smin$ is \Expected infinitely often it can never be set.
Thus, $p$ can never $\Finish$.
Moreover, since the execution is fair, the thread holding $\smin$'s obligation will always eventually be scheduled.
Finally, recall that a thread can only $\Expect$ a signal whose level is below the levels of the obligations held by that thread.
Holding the obligation for $\smin$, whose level is minimal in $\Sinf$, $p$ can never $\Expect$ a signal in $\Sinf$.
Applying Lemma~\ref{lem:finite-paths} again, we arrive at the desired contradiction.
\end{proof}

\subsection{Unsound \lstinline{Expect} counterexample}

Recall the unsound step-rule for $\Expect$ in Section~\ref{sec:overview} of the main text:
\begin{mathpar}
\inferH{UnsoundExpect}
  { \istate.\stateSignals(s) = (\lev,\myfalse) \\
    \lev \prec \istate.\stateObligations(\theta) }
  {\Expect\spac \theta\spac s\spac \degree, \istate \hstepi{\theta} \TT, \istate : \stateCallPerms[\theta \updmsetunion \set{\degree}]}
\end{mathpar}
Figure~\ref{fig:unsound-livelock} gives a counterexample using  \ruleref{UnsoundExpect} i.e.\ a program which admits an infinite fair execution.

The left (main) thread holds lock $x$ while the right thread (forkee) holds lock $y$.
Both threads try to acquire the other thread's lock before releasing their own,  resulting in a \emph{livelock}.

Notice that the issue is not only that \ruleref{UnsoundExpect} is allowed to pick the degree at which to generate a call permission arbitrarily.
Additionally, the two threads could ``juggle'' obligations by spawning and destroying signals arbitrarily.
This is prevented by requiring an expect permission in order to \lstinline{/*Expect*/}.

\begin{figure}
\begin{tabular}{c}
\begin{lstlisting}
let x = ref true in
/*let l_1 = ref NewSignal 0 in*/
let y = ref true in
/*let l_2 = ref NewSignal 0 in*/
\end{lstlisting}
\end{tabular}
\begin{tabular}{r||l}
\begin{lstlisting}
acquire y
  /*(λ.SetSignal !l_1; Expect !l_2 δ;
   l_1 := NewSignal 0) (λ.())*/
release x ()
\end{lstlisting} &
\begin{lstlisting}
acquire x
  /*(λ.SetSignal !l_2; Expect !l_1 δ;
   l_2 := NewSignal 0) (λ.())*/
release y ()
\end{lstlisting} 
\end{tabular}
\caption{Livelocked program admitting infinite fair executions given the unsound step-rule for \lstinline{Expect}.
  Here we assume that \lstinline{acquire} is the simple \spinlock (hence the degree $\degree[0]$).}%
\label{fig:unsound-livelock}
\end{figure}

\section{Proving safety for \heaplanglt programs}
\subsection{Weakest precondition}

We adapt the weakest precondition of Iris 4.0~\cite{DBLP:journals/pacmpl/SpiesGTJKBD22,iris-technical-reference} in Figure~\ref{fig:adapted-weakestpre}.
Important changes:
\begin{itemize}
  \item We assume a fixed state interpretation (Figure~\ref{fig:state-interpretation}) suitable for \heaplanglt.
  \item We always prove non-stuckness.
  \item Our definition of reducibility $\red(\expr,\istate,\theta)$ (Figure~\ref{fig:adapted-red}) accepts $\mathsf{None}$, which we use to support the $\abort$ command.
    Notice that $\red(\expr,\istate,\theta)$ admits $\abort$ and that $\mynone$---the expression stepped to by $\abort$---vacuously satisfies the weakest precondition (see also the proof rule of $\abort$ in Figure~\ref{fig:heaplanglt-extensions-proof-rules}).
\end{itemize}

\begin{figure}
\begin{align*}
  \textdom{wp}(\stateinterp, \pred_F, \stuckness) \eqdef{}&
    \setlength\fboxsep{1pt}
    \MU \textdom{wp\any rec}. \Lam \mask, \expr\colorbox{greenbg}{$, \theta$}, \pred. \\
        & (\Exists\val. \toval(\expr) = \val \land \pvs[\mask] \pred(\val)) \lor {}\\
        & \Bigl(\toval(\expr) = \bot \land \All \istate, n_s, \vec\obs, \vec\obs', n_t. \stateinterp(\istate, n_s, \vec\obs \dplus \vec\obs', n_t) \vsW[\mask][\emptyset] {}\\
        & \setlength\fboxsep{1pt}
          \qquad (s = \NotStuck \Ra \red(\expr, \istate\colorbox{greenbg}{$, \theta$})) * \All \expr', \istate', \vec\expr. (\expr, \istate \stepi[\vec\obs]{\colorbox{greenbg}{$\theta$}} \expr', \istate', \vec\expr) \wand \laterCredit{(n_\rhd(n_s)+1)} \wand {}\\
        & \setlength\fboxsep{1pt}
          \qquad\qquad (\pvs[\emptyset]\later\pvs[\emptyset])^{n_\rhd(n_s)+1} \pvs[\emptyset][\mask]\stateinterp(\istate', n_s + 1, \vec\obs', n + |\vec\expr|) * \textdom{wp\any rec}(\mask, \expr', \pred\colorbox{greenbg}{$, \theta$}) * {}\\
        & \setlength\fboxsep{1pt}
          \qquad\qquad\qquad \Sep_{(\theta', \expr'') \in \vec\expr} \textdom{wp\any rec}(\top, \expr'', \pred_F\colorbox{greenbg}{$, \theta'$})\Bigr) \\
  \setlength\fboxsep{1pt}
  \wpre[\stateinterp;\pred_F]\expr[\stuckness;\mask\colorbox{greenbg}{$, \theta$}]{\Ret\val. \prop} \eqdef{}&
    \setlength\fboxsep{1pt}
    \textdom{wp}(\stateinterp,\pred_F,\stuckness)(\mask, \expr\colorbox{greenbg}{$, \theta$}, \Lam\val.\prop)
\end{align*}
\caption{We adapt Iris' weakest precondition~\cite[p.\ 35]{iris-technical-reference}, specializing it to \heaplanglt. We add a notion of thread ids that is built into the operational semantics to support our auxiliary constructs.
  We fix $\stuckness = \NotStuck$, $n_\rhd(n_s) = 0$, and the state interpretation $\stateinterp$ of Figure~\ref{fig:state-interpretation}.}%
\label{fig:adapted-weakestpre}
\end{figure}

\subsection{State interpretation}

Our fixed state interpretation (Figure~\ref{fig:state-interpretation}) to tie the state of Figure~\ref{fig:heaplanglt-state-appendix} to the proof state follows the general pattern of tying the physical state to the ghost state of Iris.
Iris' authoritative camera~\cite[4.8]{iris-technical-reference} is used to store the authoritative fragment in the state interpretation, itself kept by the weakest precondition.
This still permits fragments to be given out for modular proofs.
Nevertheless, the authoritative part in the state interpretation makes sure that at every step the ghost state matches the physical state.

\begin{figure}
  \begin{align*}
    \stateinterp(\istate, n_s, \vec\obs, n_t) \eqdef
             & \ownGhost{\gamma_{\stateHeap}}{\authfull \istate.\stateHeap : \authm(\Loc \fpfn (\fracm \times \agm(\Val))} \ast{} & \text{cf.\ \citet{DBLP:journals/jfp/JungKJBBD18}} \\
             & \ownGhost{\gamma_{\textsc{StepCount}}}{\authfull n_s : \monom(\nat)} \ast{} & \text{cf.\ \citet{DBLP:conf/pldi/MatsushitaDJD22}}\\
             & \left(\begin{array}{ll}
                 \exists \Pi. & \ownGhost{\stateProphs}{\authfull \Pi} \ast \dom(\Pi) = \istate.\stateProphs \ast{} \\
                              & \forall\{p \gets \mathit{vs}\} \in \Pi.\, \mathit{vs} = \mathsf{filter}(p,\vec{\obs})) \ast{}
                \end{array}\right)
                \ast{} & \text{\citet{DBLP:journals/pacmpl/JungLPRTDJ20}} \\
             & \text{Our extensions for \heaplanglt:} \\
             & \ownGhost{\gamma_{\stateHeapA}}{\authfull \istate.\stateHeapA : \authm(\Loc \fpfn (\fracm \times \agm(\Val)))} \ast{} \\
             & \ownGhost{\gamma_{\stateSignals}}{\authfull \istate.\stateSignals : \authm(\Siglocdom \fpfn \exm(\Levdom \times \mathbb{B}))} \ast{} \\
             & \ownGhost{\gamma_{\stateObligations}}{\authfull \istate.\stateObligations : \authm(\TIddom \fpfn \exm(\multisetm(\Siglocdom \times \Levdom)))} \ast{} \kern-30ex{} \\
             & \ownGhost{\gamma_{\stateCallPerms}}{\authfull \istate.\stateCallPerms : \authm(\TIddom \fpfn \multisetm(\Degdom))} \ast{} \\
             & \ownGhost{\gamma_{\stateExpectPerms}}{\authfull \istate.\stateExpectPerms : \authm(\TIddom \fpfn \setm(\Siglocdom \times \Degdom))}
  \end{align*}
  \caption{Extended state interpretation.}%
  \label{fig:state-interpretation}
\end{figure}

\begin{figure}
  \begin{mathpar}
    \inferrule
      { \Exists \vec\obs, \expr_2, \istate_2, \vec\expr. \expr,\istate \stepi[\vec\obs]{\theta} \expr_2,\istate_2,\vec\expr }
      { \red(\expr, \istate, \theta) }

    \inferrule
      { \expr,\istate \stepi[\nil]{\theta} \mynone }
      { \red(\expr, \istate, \theta) }
  \end{mathpar}
  \caption{Adaption of $\red(\expr,\istate)$ to admit the $\abort$ command.}%
  \label{fig:adapted-red}
\end{figure}

\subsection{Proof rules for \heaplanglt extensions}

Figure~\ref{fig:heaplanglt-extensions-proof-rules} gives proof rules for our extensions of \heaplang.
They directly reflect the step rules of the operational semantics of these extensions (main text Figure~\ref{fig:heaplanglt-extensions-appendix}, reproduced in this appendix in Figure~\ref{fig:heaplanglt-extensions-appendix}) and can be derived from the weakest precondition (Figure~\ref{fig:adapted-weakestpre}).

\begin{figure}
  \begin{align*}
    \obligations{O}[\theta] \ast{} \quad{} \\
      \later(\forall s.\, \signal{s}{\lev}{\myfalse} \ast \obligations{O \cup \set{(s,\lev)}}[\theta] \wand \Phi(s)) \proves{}& \wpre {\langkw{NewSignal}\spac\theta\spac \lev} {\Phi} \\
    \signal{s}{\lev}{\myfalse} \ast \obligations{O \cup \set{(s,\lev)}}[\theta] \ast{} \quad{}\\
      \later(\signal{s}{\lev}{\mytrue} \ast \obligations{O}[\theta] \wand \Phi())) \proves{}& \wpre {\langkw{SetSignal}\spac \theta\spac s} {\Phi} \\
    \mathit{CP}_{\theta} \ast \obligations{O}[\theta] \ast \mathit{EP}_{\theta} \ast
    \quad{}\\
    \later\Bigl(\forall \theta'.\, \obligations{O\setminus\overline{s}}[\theta] \ast \obligations{\overline{s}}[\theta'] \ast \mathit{CP}_{\theta} \ast \mathit{CP}_{\theta'} \ast \mathit{EP}_{\theta} \ast \mathit{EP}_{\theta'}
    \Bigr) \vsW\quad{} \\
       \later\Phi() \ast \later\wpre e {\obligations{\emptyset}[\theta']} \proves{}& \wpre {\langkw{fork}\spac\overline{s}\spac\spac \expr} {\Phi} \\
    \TRUE \ast \later\Phi(\theta) \proves{}& \wpre {\langkw{CurrentThread}} [\theta] {\Phi} \\
    \obligations{\emptyset}_{\theta} \ast \later\Phi \proves{}& \wpre {\Finish}[\theta] {\Phi} \\
    \bswpre e {x.\, \Phi} \proves{}& \wpre {\langle\expr\rangle} {x.\, \Phi} \\
    \callp{\degree[0]} \ast \later(\wpre {\subst{{\subst{\expr}{\lvarF}{(\MyRecE{\lvarF}{\lvar}{\expr})}}}{\lvar}{v}} {\Phi}) \proves{}& \wpre {(\MyRecV{\lvarF}{\lvar}{\expr})\spac v} {\Phi} \\
    \callp{\delta}[\theta'] \ast \later(\callp{[n \cdot \delta']}[\theta'] \wand \Phi) \proves{}& \wpre {\langkw{lower}\spac \degree\spac \langkw{to}\spac n\spac \langkw{times}\spac \degreeA \langkw{at}\spac \theta'} {\Phi} \\
    \callp{\degree}[\theta'] \ast \signal{s}{\_}{\_} \ast \degreeA < \degree \ast \later(\callp{\degree}[\theta'] \wand \waitp{s}{\degreeA}[\theta'] \wand \Phi) \proves{}& \wpre {\langkw{NewExpectPerm}\spac \theta'\spac s\spac \degree\spac \degreeA} {\Phi} \\
    \signal{s}{\lev}{\myfalse} \ast \obligations{O}[\theta'] \ast \lev \prec O \ast \waitp{s}{\degree}[\theta'] \ast \later(\callp{\degree}[\theta'] \wand \Phi) \proves{}& \wpre {\langkw{Expect}\spac\theta'\spac s\spac \degree} {\Phi} \\
    \TRUE \proves{}& \wpre {\abort} {\FALSE}
  \end{align*}
  \caption{Proof rules for \heaplanglt's extensions.
    }%
  \label{fig:heaplanglt-extensions-proof-rules}
\end{figure}

\subsection{Proving $\texttt{deref}$}

Figure~\ref{fig:deref-proof} proves $\texttt{deref}$ against the specification given in Section~\ref{sec:lat} of the main text.

\begin{figure}
\begin{mathpar}
  \inferrule*[Right=\ruleref{hoare-def}]
    { 
      \inferrule*[rightstyle=\rmfamily,Right=pure step]
        {
        \inferrule*[Right=\fakeruleref{wp-atomic}]
          {
            \inferrule*[rightstyle=\rmfamily,Right=apply $\vsopen$]
              {
                \inferrule*[rightstyle=\rmfamily,Right=\textsc{fup-frame} (not shown)]
                  {
                    \inferrule*[rightstyle=\rmfamily,Right=\textsc{wp-frame} (not shown)]
                      { \square \left( (\ell \mapsto w) \wand  \wpre{\deref\ell} [\emptyset] {v.\, \Phi(v) \ast \ell \mapsto w} \right) }
                      {\inferrule*[rightstyle=\rmfamily,Right=\textsc{wp-mono} (not shown)]
                          {\square \left( (\ell \mapsto w \ast \vsclose) \wand  \wpre{\deref\ell} [\emptyset] {v.\, \Phi(v) \ast \ell \mapsto w \ast \vsclose} \right)}
                          {\square \left( (\ell \mapsto w \ast \vsclose) \wand  \wpre{\deref\ell} [\emptyset] {v.\, \Phi(v) \ast \pvs[\emptyset][\top] \TRUE} \right)}}
                  }
                  {\square \left( (\ell \mapsto w \ast \vsclose) \wand  \wpre{\deref\ell} [\emptyset] {v.\, \pvs[\emptyset][\top] \Phi(v)} \right)}
              }
              { \square \left( (\vsopen \ast \vsclose) \wand \pvs[\top][\emptyset] \wpre{\deref\ell} [\emptyset] {v.\, \pvs[\emptyset][\top] \Phi(v)} \right) }
            }
            { \square \left( (\vsopen \ast \vsclose) \wand  \wpre{\deref\ell} {v.\, \Phi(v)} \right) }
        }
        {\square \left( (\vsopen \ast \vsclose) \wand \wpre{\texttt{deref}\spac\ell} {v.\, \Phi(v)} \right) }
    }
    { \{ \underbrace{(\pvs[\top][\emptyset] \ell \mapsto w)}_{\vsopen} \ast \underbrace{(\ell \mapsto w \vsW[\emptyset][\top] \Phi(w))}_{\vsclose} \}
      ~\texttt{deref}\spac\ell~
      \{v.\, \Phi(v) \} }
\end{mathpar}
\caption{Proof of \lstinline{deref} (main text, Section~\ref{sec:lat}) against its specification.}%
\label{fig:deref-proof}
\end{figure}

\subsection{Big-step weakest precondition}

To reason about big-steps (see Section~\ref{sec:heaplanglt-appendix}), we define the big-step weakest precondition Figure~\ref{fig:bs-weakestpre}.
This weakest precondition admits the expected reasoning rules (Figure~\ref{fig:bs-weakestpre-reasoning}).

\begin{figure}
  \begin{mathpar}
    \inferrule[BigStep]
    { (\expr, \istate) \stepiminusstar[\nil]{\theta} (v, \istate', \nil) }
    { (\expr, \istate) \Downarrow_{\theta} (v,\istate') }
  \end{mathpar}
  \caption{Definition of big-step relation as finite number (${}^{*}$) of modified thread steps (${}_{\textlog{t-}}$) to a value by the same thread.
    The modified thread-step relation $\stepiminus[\nil]{\theta}$ uses duplicates of the regular thread- and head-step relations but disallows forks and atomic blocks.
  }%
  \label{fig:big-step}
\end{figure}

\begin{figure}
\begin{align*}
    \bswp \eqdef{}& \lambda \mask, \expr, \pred. \All \istate, n_s, \vec\obs, n_t. \stateinterp(\istate, n_s, \vec\obs, n_t) \vsW[\mask][\mask] {}\\
        & \Exists \istate', v. \ulcorner \expr, \istate \Downarrow v, \istate' \urcorner \ast \stateinterp(\istate', n_s, \vec\obs, n_t) \ast \pred(v) {}\\
    \bswpre\expr[\mask]{\Ret\val. \prop} \eqdef{}& \bswp(\mask,\expr,\Lam\val.\prop)
\end{align*}
\caption{Our weakest precondition for big-step evaluation.
  Unlike Iris' general weakest precondition~\cite[p.~35]{iris-technical-reference}, the big-step weakest precondition is specialized to \heaplanglt and in particular a fixed state interpretation. Moreover, it directly uses the big-step relation defined on a subset of \heaplanglt that notably excludes the fork command. Finally, \bswp does not allow stuck computations.
}%
\label{fig:bs-weakestpre}
\end{figure}

\begin{figure}
\begin{mathpar}
\infer[Big-step-atomic]
{}{\pvs[\mask_1][\mask_2] \bswpre\expr[\mask_2]{\Ret\var. \pvs[\mask_2][\mask_1]\prop}
 \proves \bswpre\expr[\mask_1]{\Ret\var.\prop}}

\infer[Big-step-bind]
{\text{$\lctx$ is a context}}
{\bswpre\expr[\mask]{\Ret\var. \bswpre{\lctx(\ofval(\var))}[\mask]{\Ret\varB.\prop}} \proves \bswpre{\lctx(\expr)}[\mask]{\Ret\varB.\prop}}

\infer[Big-step-AtomicBlock]
{}
{\bswpre\expr[\mask]{\Ret\var.\prop}\proves \wpre{\langle\expr\rangle}[\mask]{\Ret\var.\prop}}
\end{mathpar}
\caption{Reasoning rules for $\textdom{wp}^{\Downarrow}$. Here, ``context'' really refers to a list of \heaplanglt-specific context items (cf.~\cite[p.~36]{iris-technical-reference}).}%
\label{fig:bs-weakestpre-reasoning}
\end{figure}

\subsection{Inlined \spinlock}

Figure~\ref{fig:unsound-fine-inlined} shows a simple \heaplanglt program with the bodies of \lstinline{acquire} and \lstinline{release} inlined.

\begin{figure}
\begin{tabular}{c}
\begin{lstlisting}
$\{ \obligations{\emptyset}[1] \ast \callp{1} \}$
let lk = ref true in (* $\cursivecomment{init.\ locked}$ *)
/*let s = NewSignal 0 in
NewExpectPerm s 1 0;*/
$\{ \obligations{\set{(s,0)}} \ast \lockvar \mapsto \mytrue \ast \signal{s}{0}{\myfalse} \ast \waitp{s}{0}[1] \ast \extok{} \}$
$\knowInv{\namesp}{I \eqdef \exists b.\, \lockvar \mapsto b \ast \left((b = \mytrue \ast (\signal{s}{0}{\myfalse} \lor \extok{})) \lor (b = \myfalse \ast R)\right)}$
$\{ \knowInv{\namesp}{I} \ast R \ast \extok{} \}$
/*(* $\cursivecomment[myghostcode]{forker: keep obligation}$ *)*/
\end{lstlisting}
\end{tabular}\\
\begin{tabular}{c||c}
\begin{lstlisting}
$\{ \obligations{\set{(s,0)}}[1] \ast R \ast \knowInv{\namesp}{I} \}$
(* $\cursivecomment{inlined}$ release lk /*κ*/ *)
  $\{ \obligations{\set{(s,0)}}[1] \ast R \ast I \}_{\emptyset}$
/*⟨*/lk := false/*; $\color{myghostcode}\overbrace{\SetSignal{\texttt{s}}}^{\aalpha}$⟩*/ 
  $\{ \obligations{\emptyset}[1] \ast \lockvar \mapsto \myfalse \ast R \}_{\emptyset}$
  $\{ \obligations{\emptyset}[1] \ast I \}_{\emptyset}$
$\{ \obligations{\emptyset}[1] \ast \knowInv{\namesp}{I} \}$

\end{lstlisting} & 
\begin{lstlisting}
$\{ \obligations{\emptyset}[2] \ast \waitp{s}{0}[2] \ast \extok{} \ast \knowInv{\namesp}{I} \}$
(* $\cursivecomment{inlined}$ acquire lk /*η κ*/ *)
while not
$\{ \obligations{\emptyset}[2] \ast \waitp{s}{0}[2] \ast \extok{} \ast \knowInv{\namesp}{I} \}$
    $\{ \obligations{\emptyset}[2] \ast \waitp{s}{0}[2] \ast \extok{} \ast I \}_{\emptyset}$
   /*⟨if !lk then $\color{myghostcode} \overbrace{\Expect[\texttt{s}]}^{\abeta}$ else $\color{myghostcode} \overbrace{\texttt{()}}^{\aalpha}$;*/
    CAS lk false true/*⟩*/
    $\left\{
       \begin{array}{l}
         v.\, \lockvar \mapsto \mytrue \ast
         \left(
           \begin{array}{l}
             (v = \myfalse \ast b = \mytrue \ast \signal{s}{0}{\myfalse} \ast \callp{0}[2]) \lor{} \\
             (v = \mytrue \ast b = \myfalse \ast R)
           \end{array}
         \right) \\
         {} \ast \waitp{s}{0}[2] \ast \extok{}
       \end{array}
      \right\}_{\emptyset}$
$\{v.\, \obligations{\emptyset}[2] \ast \waitp{s}{0}[2] \ast (\overbrace{v = \mytrue \ast R}^{\text{\rmfamily \itshape acquired}} \lor (\overbrace{v = \myfalse \ast \callp{0} \ast \extok{}}^{\text{\rmfamily\itshape retry}})) \ast \knowInv{\namesp}{I} \}$
do () done
$\{ \obligations{\emptyset}[2] \ast \knowInv{\namesp}{I} \ast R \}$
\end{lstlisting}
\end{tabular}
\caption{Inlined \spinlock \lstinline{acquire} and \lstinline{release}.
  Recall that compare-and-swap (\lstinline{CAS}) returns true if the swap is succesfully executed and false otherwise.).
  We use the shorthand $\extok{} \eqdef \ownGhost{\gamma}{\exinj()}$ to represent an exclusive ghost token.}%
\label{fig:unsound-fine-inlined}
\end{figure}

\section{Motivating example}

We give a proof outline for the instrumented (\heaplanglt) version of the motivating example.
The explanation proceeds in four steps:
\begin{enumerate}
  \item We show the proof outline for the \spinlock operations against the general unfair specification given in the paper.
  \item We show the client's invariant and highlight key aspects of it.
  \item We show proof outlines for the client's auxiliary functions $\abeta$ and $\aalpha$.
  \item We explain how the invariant and thread-local resources let us derive the necessary total correctness logically atomic updates to invoke \lstinline{acquire} and \lstinline{release}.
\end{enumerate}

\subsection{\Spinlock proof}

First, Figures~\ref{fig:app-spinlock-new}, \ref{fig:app-spinlock-release} and~\ref{fig:app-spinlock-acquire} show proof outlines for the instrumented \spinlock against our total correctness unfair lock specification.

\begin{figure}
\end{figure}

\begin{figure}
  \begin{subfigure}{.4\textwidth}
  $L(x,b) \eqdef x \mapsto b$
  \caption{Lock predicate for \spinlock.}
  \end{subfigure}
  \begin{subfigure}{.4\textwidth}
  \begin{tabular}{c}
\begin{lstlisting}
let new_lock = μr _ _ /*_*/.
  $\{ \TRUE \}$
  ref false
  $\{\ell.\, \ell \mapsto \myfalse \}$
  $\{\ell.\, L(\ell, \myfalse) \}$
\end{lstlisting}
  \end{tabular}
  \caption{Create new \spinlock.}%
  \label{fig:app-spinlock-new}
  \end{subfigure}
  \caption{\Spinlock representation and creation.}
\end{figure}

\begin{figure}
  \begin{tabular}{c}
\begin{lstlisting}
let release = μr lk /*κ*/.
  $\{ \callp{\sldeg} \ast \tatomicupdate \}~\text{where}~\tatomicupdate = {}_{\sldeg}\langle L(\lockvar,\mytrue) \midalpha L(\lockvar,\myfalse) \rangle_{\emptyset}^{\lev}$
  ⟨
  $\{ \tatomicupdate \}$
    $\{ L(\lockvar, \mytrue) \ast (\bswpre{{\color{myghostcode}\aalpha\spac\TT}}{L(\lockvar,\myfalse) \vsW[\emptyset][\top] \Phi) }) \}_{\emptyset}$
    /*κ ();*/
    $\{ L(\lockvar,\mytrue) \ast \obligations{O} \ast \lev \prec O \ast (L(\lockvar,\myfalse) \vsW[\emptyset][\top] \Phi) \}_{\emptyset}$
    $\{ \lockvar \mapsto \mytrue \ast \obligations{O} \ast (L(\lockvar,\myfalse) \vsW[\emptyset][\top] \Phi) \}_{\emptyset}$
    lk := false
    $\{ \lockvar \mapsto \myfalse \ast \obligations{O} \ast (L(\lockvar,\myfalse) \vsW[\emptyset][\top] \Phi) \}_{\emptyset}$
    $\{ L(\lockvar,\myfalse) \ast \obligations{O} \ast (L(\lockvar,\myfalse) \vsW[\emptyset][\top] \Phi) \}_{\emptyset}$
    $\{ \Phi \ast \obligations{O} \}_{\emptyset}$
  ⟩
  $\{ \Phi \ast \obligations{O} \ast \lev \prec O \}$
\end{lstlisting}
  \end{tabular}
  \caption{Release \spinlock.}%
  \label{fig:app-spinlock-release}
\end{figure}

\begin{figure}
  \begin{tabular}{c}
\begin{lstlisting}
let acquire = μr acq lk /*η*/ /*κ*/.
  $\{ \callp{\sldeg} \ast \obligations{O} \ast \lev \prec O \ast \tatomicupdate \}~\text{where}~ \tatomicupdate = {}_{\sldeg}^{\lev}\langle b \eventuallybeta\set{\myfalse}.\, L(\lockvar,b) \midalpha L(\lockvar,\mytrue) \ast b = \myfalse \rangle_{\emptyset}$
  $\{ \callp{\sldeg} \ast \obligations{O} \ast \lev \prec O \ast \tatomicupdate \ast \later((\callp{\sldeg} \ast \tatomicupdate) \wand \wpre{\texttt{acq}\spac\texttt{lk}\spac\abeta\spac\aalpha}{\Phi}) \}$
  /*lower ⊤SL into 1 times 0*/
  $\{ \callp{\degree[0]} \ast \obligations{O} \ast \lev \prec O \ast \tatomicupdate \ast \later((\callp{\sldeg} \ast \tatomicupdate) \wand \wpre{\texttt{acq}\spac\texttt{lk}\spac\abeta\spac\aalpha}{\Phi})  \}$
  if
  ⟨
  $\{ \tatomicupdate \}$
    $\left\{\begin{array}{l} \exists b.\, L(\lockvar,b) \ast{}\\
      (((b = \myfalse \ast \obligations{O}) \wand \bswpre{{\color{myghostcode}\aalpha\spac\TT}}[\emptyset]{(L(\lockvar,\mytrue) \vsW[\emptyset][\top] \Phi)}) \land{}\\
      ((b = \mytrue \ast \obligations{O} ) \wand \bswpre{{\color{myghostcode}\abeta\spac\TT}}[\emptyset]{\callp{\sldeg} \ast \obligations{O} \ast (L(\lockvar,b) \vsW[\emptyset][\top] \tatomicupdate)}))
    \end{array}\right\}_{\emptyset}$
    /*(if !lk then
      η ()
    $\{ \callp{\sldeg} \ast \obligations{O} \ast (L(\lockvar,\mytrue) \vsW[\emptyset][\top] \tatomicupdate) \}_{\emptyset}$
    else
      κ ()
    $\{ L(\lockvar,\mytrue) \vsW[\emptyset][\top] \Phi \}_{\emptyset}$
    );*/
    CAS lk false true
  $\{ v.\, (v = \mytrue \ast \Phi) \lor (v = \myfalse \ast \callp{\sldeg} \ast \obligations{O} \ast \tatomicupdate) \}$
  ⟩
  then
    $\{ \Phi \}$
    () (* $\cursivecomment{return}$ *)
  else
    $\{ \callp{\degree[0]} \ast \callp{\sldeg} \ast \obligations{O} \ast \lev \prec O \ast \tatomicupdate \ast ((\callp{\sldeg} \ast \tatomicupdate) \wand \wpre{\texttt{acq}\spac\texttt{lk}\spac\abeta\spac\aalpha}{\Phi})  \}$
    acq lk /*η*/ _ /*κ*/
    $\{ \Phi \}$
\end{lstlisting}
  \end{tabular}
  \caption{Acquire \spinlock.
    }%
  \label{fig:app-spinlock-acquire}
\end{figure}

\subsection{Client invariant}

\begin{figure}
  \begin{align*}
    I(lk,f,l_2,l_3,\gnameowner,\gnamehandoff,\gname_{l_2},\gname_{l_3}) ={}& \exists b,d,s_2,s_3.\, L(\lockvar, b) \ast f \mapsto d \ast l_2 \mapsto s_2 \ast l_3 \mapsto s_3 \ast{} \\
    \mathsf{Progress}\qquad & \left(\begin{array}{l}
             (b = \myfalse \ast \ownerNone) \lor{}\\
             \left(\begin{array}{l}(b = \mytrue \ast \left(\begin{array}{l}(\ownerLM \ast \signal{s_2}{0}{\myfalse} \ast \shot{l_2}{s_2}) \lor{}\\ (\ownerR \ast \signal{s_3}{0}{\myfalse} \ast \shot{l_3}{s_3})\end{array}\right)\end{array}\right)
           \end{array}\right) \ast {} \\
    \mathsf{Handoff}\qquad & \left(\begin{array}{l}
             d = \mytrue \lor{} \\
             \left(d = \myfalse \ast \left(\extok{\namehandoff} \lor \left(\ownerLM \ast \obligations{\set{(s_2,0)}}[2]\right)\right)\right)
           \end{array}\right)
  \end{align*}
  \caption{%
    Inspired by the appendix of~\citet{DBLP:journals/pacmpl/SpiesGTJKBD22}, we use the shorthands
     $\extok{name}   \eqdef \ownGhost{\gamma_{\mathrm{name}}}{\exinj()}$,
     $\pending{name} \eqdef \ownGhost{\gamma_{\mathrm{name}}}{\mathsf{pending}}$, and
     $\shot{name}{v}    \eqdef \ownGhost{\gamma_{\mathrm{name}}}{\mathsf{shot}(v)}$.
     The exclusive and one-shot cameras are described by \citet{DBLP:journals/jfp/JungKJBBD18} and in Iris' technical reference~\cite{iris-technical-reference}.
    Thread~1 is the left thread, thread~2 the middle thread and thread~3 the right thread.
   }%
   \label{fig:motivating-example-client-inv}
\end{figure}

Figure~\ref{fig:motivating-example-client-inv} shows the (client) invariant we use to verify the instrumented motivating example.

The protocol enforced by the invariant encompasses two aspects controlled by $b$ and $d$:
\begin{itemize}
  \item $\mathsf{Progress}$: When the lock is held ($b = \mytrue$), there exists an unset signal to justify busy-waiting.
    Specifically, if the left or middle thread own the lock ($\ownerLM$) the middle thread must hold the obligation for signal $s_2$.
    If the right thread owns the lock $\ownerR$ it must hold the obligation for signal $s_3$.
  \item $\mathsf{Handoff}$: When the flag \lstinline{f} has been unset ($d = \myfalse$) the middle thread can complete the handoff between itself and the left thread---obtaining ownership.
\end{itemize}

This invariant uses the shorthand $\extok{} \eqdef \ownGhost{\gamma}{\exinj()}$ for exclusive tokens again.
Moreover, we use the oneshot camera~\cite{DBLP:journals/jfp/JungKJBBD18} with the shorthands
\( \pending{} \eqdef \ownGhost{\gamma}{\mathsf{pending}} \), and
\( \shot{}{v} \eqdef \ownGhost{\gamma}{\mathsf{shot}(v)} \).
The exclusive resource $\pending{}$ is required to set the value to $v$ and obtain the persistent resource $\shot{}{v}$.
This ensure that \lstinline{/*l_1*/} and \lstinline{/*l_2*/} can only be set once to a signal.

\subsection{Client-specific auxiliary code}

\begin{figure}
\begin{tabular}{c}
\begin{lstlisting}
$\{ l_2 \mapsto \_ \ast \obligations{\emptyset}[2] \}_{\emptyset}^{\Downarrow}$
l_2 := NewSignal 2 0
$\{ l_2 \mapsto s_2 \ast \obligations{(s_2,0}[2] \ast \signal{s_2}{\myfalse}{0} \}_{\emptyset}^{\Downarrow}$
\end{lstlisting}
\end{tabular}
\caption{Big-step weakest precondition of the left thread's $\aalpha$.}%
\label{fig:app-aalpha-left}
\end{figure}

\begin{figure}
\begin{tabular}{c}
{
\def\sldeg{\degree_3}
\begin{lstlisting}
$\{ \signal{s_3}{0}{\myfalse} \ast l_3 \mapsto s_3 \ast \startedvar \mapsto b \ast ((b = \myfalse \ast \callp{\degree_2}) \lor (b = \mytrue \ast \waitp{s_3}{\sldeg})) \}_{\emptyset}^{\Downarrow}$
if not !started then (
$\{ \startedvar \mapsto b \ast \callp{\degree_2} \}_{\emptyset}^{\Downarrow}$
  NewExpectPerm !l_3 δ_2 δ_3;
  started = true
$\{ \startedvar \mapsto \mytrue \ast \waitp{s_3}{\sldeg} \}_{\emptyset}^{\Downarrow}$
);
$\{ \signal{s_3}{\myfalse}{0} \ast \waitp{s_3}{\sldeg} \ast l_3 \mapsto s_3 \ast \startedvar \mapsto \mytrue \}_{\emptyset}^{\Downarrow}$
Expect !l_3
$\{ \callp{\sldeg} \ast \signal{s_3}{\myfalse}{0} \ast \waitp{s_3}{\sldeg} \ast l_3 \mapsto s_3 \ast \startedvar \mapsto \mytrue \}_{\emptyset}^{\Downarrow}$
\end{lstlisting}
}
\end{tabular}
\caption{Big-step weakest precondition of left thread's \abeta.}%
\label{fig:app-abeta-left}
\end{figure}

\begin{figure}
\begin{tabular}{c}
\begin{lstlisting}
$\{ l_2 \mapsto s_2 \ast \obligations{(s_2,0)} \ast \signal{s_2}{\myfalse}{0} \}_{\emptyset}^{\Downarrow}$
SetSignal !l_2
$\{ l_2 \mapsto s_2 \ast \obligations{\emptyset} \ast \signal{s_2}{\mytrue}{0} \}_{\emptyset}^{\Downarrow}$
\end{lstlisting}
\end{tabular}
\caption{Big-step weakest precondition of \aalpha (middle thread)}%
\label{fig:app-aalpha-middle}
\end{figure}

\begin{itemize}
  \item Figure~\ref{fig:app-aalpha-left} shows the proof outline for the auxiliary code of the left thread's $\aalpha$ (\lstinline{acquire}).
  \item Figure~\ref{fig:app-abeta-left} shows the proof outline for the auxiliary code of the left thread's $\abeta$ (\lstinline{acquire}).
  \item Figure~\ref{fig:app-aalpha-middle} show the proof outline for the auxiliary of the middle thread's $\aalpha$ (\lstinline{release}).
    This is the $\aalpha$ for the release matching the lock acquisition by the left thread.
\end{itemize}

Notice that these do not yet correspond to the big-step weakest preconditions required for the total correctness logically atomic update ($\tatomicupdate$).
What is missing are the linear view shift to close the invariant.

\subsection{Deriving the total correctness logically atomic update}

Recall that $\tatomicupdate$ is defined as greatest fixpoint.
We use \ruleref{gfp-intro} to introduce it.
\begin{mathpar}
  \inferH{gfp-intro}
    {P \proves F(P)}
    {P \proves \mathsf{gfp}(P.\, F(P))}
\end{mathpar}
Thus, for \lstinline{acquire} we're looking for \emph{wait invariant} resource $\waitinv$ such that:
\begin{mathpar}
  \inferrule*[Right=\ruleref{gfp-intro}]
    {\waitinv \proves \pvs[\top\setminus\mask][\emptyset] \exists b.\, \later L(\lockvar,b) \ast{} \\
    \left(\begin{array}{l}
      ((b = \mytrue \ast \obligations{O}) \wand \bswpre {\abeta\spac\TT}[\emptyset]{\callp{\sldeg} \ast \obligations{O} \ast (L(\lockvar,\mytrue) \vsW[\emptyset][\top\setminus\mask] \waitinv)}) \land{} \cr
      ((b = \myfalse \ast \obligations{O}) \wand \bswpre {\aalpha\spac\TT} [\emptyset] {(L(\lockvar,\mytrue) \vsW[\emptyset][\top\setminus\mask] \Phi)})
    \end{array}\right)
    }
    {\waitinv \proves {}_{\sldeg}^{\lev;O}\langle b \eventuallybeta \set{\myfalse}.\, L(\lockvar,b) \midalpha L(\lockvar,\mytrue) \ast b = \myfalse \vs \Phi \rangle_{\mask}}
\end{mathpar}

\subsubsection{Left thread \lstinline{acquire}}%
\label{sec:left-thread-acquire}

The client picks $\Phi$ and $\waitinv$ as follows:
\begin{align*}&
  \Phi \eqdef \obligations{\emptyset} \ast \obligations{\set{(s_2,0)}}[2] \ast \ownerLM \\&
  \waitinv \eqdef \knowInv{\namesp}{I} \ast \pending{l_2} \ast \obligations{\emptyset}[2] \ast (\exists c.\, \startedvar \mapsto c \ast ((c = \myfalse \ast \callp{\degree_2}) \lor (c = \mytrue \ast \waitp{s_3}{\degree_3}))
\end{align*}

Thus we derive \lstinline{acquire}'s $\tatomicupdate$ as follows:
\begin{itemize}
  \item Atomic precondition: $\knowInv{\namesp}{I}$ is part of $\waitinv$.
    Opening the invariant we obtain (among other resources) $\later \exists b.\, L(\lockvar,b)$.
    Later and exists commute since $\mathbb{B}$ is inhabited giving us $\later L(\lockvar,b)$.
  \item $\abeta$: To run $\abeta$, the module must show that $b$ is $\mytrue$. Therefore the $\mathsf{Progress}$ part of the invariant simplifies and---given that $\waitinv$ includes $\pending{l_2}$---the disjunct with $\ownerR$ must hold.
    At this point deriving the precondition of $\abeta$ (Figure~\ref{fig:app-abeta-left}) is relatively straightforward.

    The local resources of $\waitinv$ can be established from the expect permission and the fact that $\startedvar \mapsto \mytrue$.
    Moreover, restoring $I$ is easy since $I$'s resources were not modified.
    Thus, given $\later L(\lockvar,\mytrue)$ we can close the invariant again.
  \item $\aalpha$:
    To run $\aalpha$, the module must show that $b$ is $\myfalse$.
    Its precondition (Figure~\ref{fig:app-aalpha-left}) can be derived from the local resource $\obligations{\emptyset}[2]$ of $\waitinv$, and $l_2 \mapsto \_$ as given by $I$.
    Assuming the postcondition of $\aalpha$, we turn $\pending{l_2}$ into $\shot{l_2}{s_2}$, $\ownerNone$ into $\ownerLM$  and deposit these together with $\signal{s_2}{\myfalse}{0}$ to restore $I$.
    The other half of $\ownerLM$ and $\obligations{\set{(s_2,0)}}[2]$ are equal to $\Phi$.
    Thus, together with $\later L(\lockvar,\mytrue)$ it is now possible to close the invariant.
\end{itemize}

\subsubsection{Handoff between left and middle thread}%
\label{sec:left-middle-handoff}

The left thread opens the invariant, set \lstinline{f} to false, and deposits the right conjunct of $\mathsf{Handoff}$ with $d = \myfalse$.
Either the middle thread reads \lstinline{f} and discovers it to be set.
In that case the program exits completing the proof trivially.
Otherwise, it exchanges its token $\extok{\namehandoff}$ for the right part of that disjunction.

\subsubsection{Middle thread \lstinline{release}}%
\label{sec:middle-thread-release}

We do the same for $\aalpha$ of the middle thread's \lstinline{release}.
That is we derive
${}_{\sldeg}\langle L(\lockvar,\mytrue) \midalpha L(\lockvar,\myfalse) \vs \Phi \rangle_{\mask}^{\lev;O}$
with $\waitinv$ and $\Phi$ defined as follows.

\begin{align*}&
  \Phi \eqdef \TRUE \\&
  \waitinv \eqdef \knowInv{\namesp}{I} \ast \ownerLM \ast \obligations{\set{(s_2,0)}}
\end{align*}

\begin{itemize}
  \item Atomic precondition: From the invariant $I$ and separately $\ownerLM$ we can conclude that $\later L(\lockvar,\mytrue)$ must hold.
  \item $\aalpha$: Moreover, because $\ownerLM$ is part of the local resources of $\waitinv$ we can conclude that the $\ownerLM$ part of $\mathsf{Progress}$ holds.
    Therefore $\aalpha$'s precondition holds (Figure~\ref{fig:app-aalpha-middle}).
    Assuming its postcondition, we can close the invariant by combining the local fraction $\ownerLM$ with the one from the invariant.
    This allows us to obtain $\ownerNone$.
    $\ownerNone$ forces $L(\lockvar,\myfalse)$.
    Thus $\later L(\lockvar,\myfalse)$ is required to close up the invariant; as desired.
  \item Notice that $\aalpha$ gives the module $\obligations{\emptyset}$, which trivially satisfies $\lev \prec \emptyset$.
    Recall, that the obligations returned from \lstinline{release} (Section~\ref{sec:tclat-unfair}) must exactly be those returned from $\aalpha$ i.e.\ $\obligations{\emptyset}$.
\end{itemize}

\subsubsection{Right thread}

The calls to \lstinline{acquire} and \lstinline{release} can be justified in the same manner.

\subsection{Erasure of instrumented motivating example}

Figure~\ref{fig:app-instrumented-erased} shows the erasure of the instrumented motivating example.

\begin{figure}
  \begin{tabular}{c}
\begin{lstlisting}
let lk = new_lock () in
let f = ref true in
fork (
  if !f then exit;
  release lk
);
fork (
  acquire lk;
  release lk
);
acquire lk;
f := false
\end{lstlisting}
  \end{tabular}
  \caption{Erasure of the motivating example.}%
  \label{fig:app-instrumented-erased}
\end{figure}

\begin{figure}
\begin{tabular}{c}
\begin{lstlisting}
$\{ \obligations{\emptyset} \ast \callp{\degree_1} \}$
let lk = new_lock () in
let f = ref true in
/*let l_2 = ref 0 in let l_3 = ref 0 in
lower δ_1 into 4 times δ_2; lower δ_2 into 4 times δ_3; lower δ_2 into 4 times 0_δ;*/
$\left\{\begin{array}{l}
  \obligations{\emptyset} \ast L(\lockvar,\myfalse) \ast f \mapsto \mytrue \ast l_2 \mapsto 0 \ast l_3 \mapsto 0 \ast [2 \cdot \callp{\degree_2}] \ast [4 \cdot \callp{\degree_3}] \ast [4 \cdot \callp{\degree[0]}] \ast{} \\
  \ownerNone \ast \extok{\namehandoff} \ast \pending{l_2} \ast \pending{l_3}
  \end{array}\right\}$
$\{ \knowInv{\namesp}{I} \ast \extok{\namehandoff} \ast \pending{l_2} \ast \pending{l_3} \}$
fork /*∅*/ ( (* $\cursivecomment{Middle thread}$ *)
  $\{ \extok{\namehandoff} \ast \callp{\degree_3} \ast \callp{\degree[0]} \}$
  if !f then exit; (* $\cursivecomment{See Section~\ref{sec:left-middle-handoff}}$ *)
  $\{ \obligations{\set{(s_2,0)}} \ast \ownerLM \ast \callp{\degree_3} \ast \callp{\degree[0]} \}$
  $\{ \waitinv \ast \callp{\degree_3} \ast \callp{\degree[0]} \}$
  release lk (* $\cursivecomment{See Section~\ref{sec:middle-thread-release}}$ *)
    /*(λ. SetSignal !l_2)*/
  $\{ \Phi \ast \obligations{\emptyset} \}$
);
fork /*∅*/ ( (* $\cursivecomment{Right thread}$ *)
  $\{ \obligations{\emptyset} \ast \pending{l_3} \ast \callp{\degree_2} \ast [2 \cdot \callp{\degree_3}] \ast [2 \cdot \callp{\degree[0]} \}$
  /*let started = ref false in*/
  acquire lk
    /*(λ. if not !started then (NewExpectPerm !l_2 δ_2 δ_3; started := true);
      Expect !l_2)
  (λ l_3 := NewSignal л_1)*/;
  release lk
    /*(λ. SetSignal !l_3)*/
  $\{ \obligations{\emptyset} \}$
);
$\{ \obligations{\emptyset} \ast \obligations{\emptyset}[2] \ast \pending{l_2} \ast \pending{l_2} \ast \callp{\degree_2} \ast \callp{\degree_3} \ast \callp{\degree[0]} \}$
/*let started = ref false in*/
$\{ \waitinv \ast \callp{\degree_3} \ast \callp{\degree[0]} \}$
acquire lk (* $\cursivecomment{see Section~\ref{sec:left-thread-acquire}}$ *)
  /*(λ. if not !started then (NewExpectPerm !l_3 δ_2 δ_3; started := true);
    Expect !l_3)
  (λ l_2 := NewSignal 2 л_1)*/;
$\{ \Phi \}$
$\{ \obligations{\emptyset} \ast \obligations{\set{(s_2,0)}}[2] \ast \ownerLM \}$
f := false (* $\cursivecomment{See Section \ref{sec:left-middle-handoff}}$ *)
$\obligations{\emptyset}$
\end{lstlisting}
\end{tabular}

\caption{Proof outline for motivating example. $\degree_3$ is the embedding of $\sldeg$ into the client degrees (see Section~\ref{sec:motivating-example-revisited} of main text).}
\end{figure}

\section{\Ticketlock}

Figure~\ref{fig:ticketlock-extended} shows the \ticketlock's predicate, invariant and a proof outline for \lstinline{acquire} and \lstinline{release}.

\subsection{\Ticketlock predicate and invariant}

A key idea of the lock predicate $L$ and invariant $I$ are the fractional permissions split between them.
This forces the lock predicate and invariant to be kept in sync.
More specifically, in order to modify these shared points-to's, both the client's invariant (opaque to the module) and the lock invariant (opaque to the client) need to be opened, updated, and closed.

Other key details include:
\begin{itemize}
  \item Matching the \ticketlock's implementation, the lock is held iff the owner field is less than the next (available ticket) field.
    Note that this enforces the execution of the external linearization point at \lstinline{release} if threads are queued on the \ticketlock.
  \item Concretely, the ghost map $M$ keeps track of the \emph{helping exchange}~\cite{DBLP:conf/popl/JungSSSTBD15,DBLP:journals/pacmpl/SpiesGTJKBD22}.
    For each ticket $k$ and current value of the owner field $o$, $M$ holds one of:
    \begin{itemize}
      \item The total correctness logically atomic update to linearize the $k$'th \lstinline{acquire} if $k < o$.
        The later credit accompanying the total correctness logically atomic update (see $\tatomicupdate'$) allows stripping the later from (the not timeless) $\tatomicupdate$.
      \item The $k$'th linearization's $\Phi$ if $k = o$.
      \item An exclusive token picked by the $k$'th acquire so that it can take out $\Phi$.
    \end{itemize}
\end{itemize}

Finally, the given \ticketlock predicate nests the module's invariant inside the lock predicate.
This makes the lock's specification more concise, only requiring the lock predicate.
\citet{DBLP:journals/pacmpl/SpiesGTJKBD22}'s addition of later credits to \iris make this nesting quite manageable.
Notice the auxiliary dummy step taken at the beginning of \lstinline{acquire} and \lstinline{release} in Figure~\ref{fig:ticketlock-impl-proof-outline-appendix}.
These generate a later credit that let us strip the later from the invariant%
  \footnote{Invariants are not timeless.}
  nested inside the lock predicate.
Since the lock invariant can be duplicated, \lstinline{acquire} can keep a copy of the invariant.
This way checking the loop condition after \lstinline{acquire}'s linearization (giving up $L(\lockvar,\mytrue)$) is permitted.

An alternative approach is to define a persistent lock predicate $\mathit{is\_lock(\lockvar)}$---also returned at lock creation time---that is additionally required for \lstinline{acquire} and \lstinline{release}.
The machine-checked \verifast proof in the \supplementorartifact follows this approach.

\begin{figure}
\def\ownfld{\mathit{owner}} %
\def\nextfld{\mathit{next}} %
\def\heldfld{\mathit{held}} %
\def\mapnname{\mathsf{M}} %
\def\levelsmod{\Lev_{\mathsf{TL}}}
\def\gheld{\gname_{\mathsf{held}}}
\def\maskmod{\mask_{\mathsf{TL}}}
\raggedright
\begin{subfigure}{\textwidth}
\begin{tabular}{cc}
\begin{lstlisting}
let acquire lk /*η κ*/ =
  $\{ \callp{\top_{\mathsf{TL}}} \ast \obligations{O} \ast \Lev_{\mathsf{TL}} \prec O \ast \tatomicupdate \}$
  /*let _ = () in (* $\cursivecomment[myghostcode]{Get later credit}$ *)*/
  $\{ \laterCredit{1} \ast \callp{\top_{\mathsf{TL}}} \ast \obligations{O} \ast \Lev_{\mathsf{TL}} \prec O \ast \tatomicupdate \}$
  let t = (* $\cursivecomment{Ticket}$ *)
  /*⟨if !(lk.next) = lk.owner
   then (* $\cursivecomment[myghostcode]{\textbf{Acq-LP-1}}$ *) κ () 
   else ($\color{myghostcode}\textsc{fairness}$;
     lk.TAUs[!(lk.next)] := κ);*/
   FAA lk.next 1/*⟩*/
  in
  $\{ \knowInv{\mask}{I(\lockvar,\gname)} \ast t \fgmapsto[\gname_M] (\theta,O,\Lev_{\mathsf{TL}},\abeta,\aalpha,\Phi,\gname_t) \ast \ownGhost{\gname_t}{\exinj()} \}$
  while
  /*⟨if !(lk.owner != t) then η () else ();*/
   !(lk.owner) != t/*⟩*/
  do () done
  $\{ \Phi \}$
\end{lstlisting} &
\begin{lstlisting}
let release lk /*κ*/ =
  /*let _  = () in (* $\cursivecomment[myghostcode]{Get later credit}$ *)*/
  $\{ \callp{\top_{\mathsf{TL}}} \ast \tatomicupdate \ast \laterCredit{1} \}$
  /*⟨
  $\left\{\begin{array}{l}
    \lockvar.\ownfld \mapsto o \ast \lockvar.\nextfld \mapsto n \ast \gheld \mapsto \mytrue \ast{} \\
    \ldots \cursivecomment[myspecs]{Rest of I} \ast{} \\
    (L(\lockvar,\myfalse,o + 1) \vsW[\emptyset][\top\setminus\mask] \Phi)
  \end{array}\right\}_{\emptyset}$
   (* $\cursivecomment[myghostcode]{\textbf{Rel-LP}}$ *) κ ();*/
   FAA lk.owner 1/*;
  $\left\{\begin{array}{l}
    \lockvar.\ownfld \mapstofrac{1}{2} o + 1 \ast \lockvar.\nextfld \mapsto n \ast{}\\
    \gheld \mapstofrac{1}{2} \myfalse \ast \ldots{} \cursivecomment[myspecs]{Rest of I} \ast{}\\
    {}\ast \Phi
  \end{array}\right\}_{\top\setminus\mask}$
   if !(lk.owner) < !(lk.next)
   then (* $\cursivecomment[myghostcode]{\textbf{Acq-LP-2}}$ *)
  $\{ \tatomicupdate' \cursivecomment[myspecs]{for}~ k = o + 1 \}_{\top\setminus\mask}$
     lk.TAUs[!(lk.owner)] ()
  $\{ \Phi_k \ast k = o + 1 \}_{\top\setminus\mask}$
   else
  $\{ o + 1 = n ~\cursivecomment[myspecs]{outside domain of} \Sep{} \}_{\top\setminus\mask}$
    ()
  ⟩*/
  $\{ \Phi \ast \obligations{O} \ast \Lev_{\mathsf{TL}} \prec O \}$
\end{lstlisting}
\end{tabular}
\caption{\Ticketlock implementation.
  The proof outline references definitions of \ref{fig:ticketlock-predicate-invariant-appendix}.}%
\label{fig:ticketlock-impl-proof-outline-appendix}
\end{subfigure}

\begin{subfigure}{\textwidth}
\def\ownfld{\mathit{owner}} %
\def\nextfld{\mathit{next}} %
\def\heldfld{\mathit{held}} %
\def\mapnname{\mathsf{M}} %
\def\levelsmod{\Lev_{\mathsf{TL}}}
\def\gheld{\gname_{\mathsf{held}}}
\def\maskmod{\mask}
\begin{align*}
& \setlength\fboxsep{1pt}
\tatomicupdate'(\theta,O,\Lev_{\mathsf{TL}},\abeta,\aalpha,\Phi) \eqdef \obligations{O}[\theta] \ast \Lev_{\mathsf{TL}} \prec O \ast{}\\&\qquad
  \laterCredit{1} \ast \langle b,o \eventuallybeta_{o} \set{\myfalse} \times \nat.\, L(\lockvar,b,o) \midalpha L(\lockvar,\mytrue,o) \ast b = \myfalse \rangle_{\theta,\top\setminus\maskmod}^{O,\levelsmod} \\ &
L(\lockvar,b,o,\gname) \eqdef \gname \mapsto_{\square} (\gheld, \gname_M) \ast \gheld \mapstofrac{1}{2} b \ast \lockvar.\ownfld \mapstofrac{1}{2} o \ast \knowInv{\maskmod}{I(\lockvar,\gname)} \\& \setlength\fboxsep{1pt}
I(\lockvar,\gname) \eqdef \exists b,o,n,M.\,
  (o \leq n) \ast (b \leftrightarrow o < n) \ast \dom(M) = [0,n) \ast{} \\&\qquad
  \lockvar.\ownfld \mapstofrac{1}{2} o \ast \gheld \mapstofrac{1}{2} b \ast \lockvar.\nextfld \mapsto n \ast \ownGhost{\gname_M}{M} \ast{} \\&\qquad
  \Sep_{k \mapsto (\theta,O,\Lev_{\mathsf{TL}},\abeta,\aalpha,\Phi,\gname_k) \in M} \hspace{-3em} \bigl( (\tatomicupdate'(\theta,O,\Lev_{\mathsf{TL}},\abeta,\aalpha,\Phi) \ast o < k) \lor (\Phi \ast o = k) \lor  (\ownGhost{\gname_k}{\exinj()} \ast o \geq k) \bigr)
\end{align*}
\caption{Lock predicate for \ticketlock implementation. Parts highlighted in green mark additions necessary to satisfy the fair specification.}%
\label{fig:ticketlock-predicate-invariant-appendix}
\end{subfigure}
\caption{\Ticketlock implementation and proof outline.}%
\label{fig:ticketlock-extended}
\end{figure}

\section{\Cohortlock}

\subsection{$\abeta$s}

We briefly recall key details from Section~\ref{sec:cohortlock} of the main text.

\paragraph{Top-level lock (TL)}

The \cohortlock is held by the client whenever TL is held.
Thus, TL's $\abeta$ can use the client $\abeta$ whenever TL is held.
Recall that TL can be passed without competition from other cohorts MAX number of times.
This means the TL round does not change while the \cohortlock round does.
Thus, for each TL round, TL's $\abeta$ may need to call the \cohortlock client's \abeta up to MAX times.

\paragraph{Local lock (LL)}

For each local lock round, \lstinline{acquire} may be blocked on the ``client'' (here the \cohortlock implementation) if:
\begin{itemize}
  \item The TL, and thus the \cohortlock, is held by the same cohort as the LL \lstinline{acquire}: $\abeta$ can invoke the \cohortlock client's $\abeta$ to justify progress.
  \item The LL is held but TL has not yet been acquired by a member of the cohort: the LL $\abeta$ can expect the acquire signal.
  \item The LL is held, the TL has been released by a member of the cohort but LL's release has not yet occurred: the LL $\abeta$ can expect the release signal.
\end{itemize}

\subsection{Levels}

Let $\Lev_{\fairlockabbr}$ be the set of levels for a fair lock module used for the top-level lock (TL) and the local locks (LLs).
\begin{align*}
  \Lev_{\cohortlockabbr} \eqdef \set{0} \times \Lev_{\fairlockabbr} \cup \set{(1,0)} \cup \set{2} \times \Lev_{\fairlockabbr}
\end{align*}
The embedded levels are ordered lexicographically.
We use $\set{0} \times \Lev_{\fairlockabbr}$ for the TL,
$(1,0)$ for the acquire and release signals, and
$\set{2} \times \Lev_{\fairlockabbr}$ for the LLs.

\subsection{Degrees}

Let $\Deg_C^{\mathsf{eta}}$ be the client's eta domain of degrees (cf.\ Section~\ref{sec:ticketlock} of the main text) with the named degrees $\top_C^{\mathsf{eta}}$ and $\bot_C^{\mathsf{eta}}$.
We define $\Deg_\mathsf{TL}^{\mathsf{eta}}$ and $\Deg_\mathsf{LL}^{\mathsf{eta}}$ in terms of their local domain of degrees:
\begin{align*}&
  \Deg_\mathsf{TL}^{\mathsf{eta}} \eqdef \set{0} \times \Deg_C^{\mathsf{eta}} \cup \set{(1,0)} \\&\qquad
    \top_\mathsf{TL}^{\mathsf{eta}} \eqdef (1,0), \bot_{\mathsf{TL}}^{eta} \eqdef (0,\bot_{\mathsf{C}}^{\mathsf{eta}}) \\&
  \Deg_\mathsf{LL}^{\mathsf{eta}} \eqdef \set{0} \times \Deg_C^{\mathsf{eta}} \cup \set{(1,0)} \cup \set{(2,0)} \\&\qquad
    \top_\mathsf{LL}^{\mathsf{eta}} \eqdef (2,0), \bot_{\mathsf{LL}}^{eta} \eqdef (0,\bot_{\mathsf{C}}^{\mathsf{eta}})
\end{align*}

\subsection{\lstinline{alone}}

Figure~\ref{fig:ticketlock-alone-spec} shows the specification to \ticketlock's \lstinline{alone} operation.
This specification enforces that if the current owner is not alone, then the next release directly transfers ownership of the \ticketlock to the next waiting thread.

\begin{figure}
{
\begin{align*}&
  \langle L(\lockvar,\mytrue,r) \mid R,v.\, L(\lockvar,\mytrue,r) \ast R \ast \left(\begin{array}{l}
    (v = \myfalse \ast \notalone(\lockvar,r + 1)) \lor{}\\
    (v = \mytrue \ast (R \ast \notalone(\lockvar,r+1) \vs[\emptyset] \FALSE)
    \end{array}\right) \\&\qquad
    \vs \Phi(v) \ast R \rangle_{\top\setminus\mask} \wand \wpre{\texttt{alone}\spac\lockvar}{v.\, \Phi(v)} \\&
  L(\lockvar,\myfalse,r) \ast \notalone(\lockvar,r) \vs[\mask] \FALSE
\end{align*}
}
\caption{Specification of \ticketlock's \lstinline{alone}.
  It ensures two properties. First, if \lstinline{alone} returns false, the next lock release will be a direct handoff i.e.\ linearize the \lstinline{acquire} of the next thread in line.
  It follows that the client will at no point observe the lock not being held in the next round.
  Second, if \lstinline{alone} has returned false for the current round, it will not return true for that same round.}%
\label{fig:ticketlock-alone-spec}
\end{figure}

\section{\heaplanglt}%
\label{sec:heaplanglt-appendix}

The main text presents our extensions to \heaplang: the auxiliary code, and function application---instrumented to consume a call permission at every call.
The following figures complete \heaplanglt's definition.
\begin{itemize}
  \item \heaplanglt's syntax is defined in Figure~\ref{fig:heaplanglt-syntax-full} with additional syntactic sugar defined in Figure~\ref{fig:heaplanglt-syntactic-sugar}.
  \item \heaplanglt's state is defined in Figure~\ref{fig:heaplanglt-state-appendix}.
  \item \heaplanglt's pure and boxed reduction rules are defined in Figure~\ref{fig:heaplang-reduction-pure}.
  \item \heaplanglt's impure reduction steps are shown in Figure~\ref{fig:heaplang-opsem}.
  \item \heaplanglt's evaluation contexts are defined in Figure~\ref{fig:heaplanglt-evalctxs}.
  \item \heaplanglt's big-step evaluation relation is defined in Figure~\ref{fig:big-step}.
    Its definition uses copies of the regular thread-step (Figure~\ref{fig:heaplang-reduction-pure}) and head-step relations that exclude fork-statements and atomic blocks.
\end{itemize}

\begin{figure}[H]
  \input{figures/full-syntax.tex}
  \caption{\heaplanglt's syntax, mostly identical to \heaplang~\cite{iris-technical-reference}.
    Auxiliary code added by \heaplanglt is highlighted in blue.
    Note the additional arguments to \langkw{fork}, explained in Figure~\ref{fig:heaplanglt-extensions-appendix}.
    Syntactic sugar for some constructs is defined in Figure~\ref{fig:heaplanglt-syntactic-sugar}.
  }%
  \label{fig:heaplanglt-syntax-full}
\end{figure}

\begin{figure}[H]
  \input{figures/sugar.tex}
  \caption{\heaplanglt syntactic sugar.}%
  \label{fig:heaplanglt-syntactic-sugar}
\end{figure}

\begin{figure}[H]
  \input{figures/state/heaplanglt-state.tex}
  \caption{\heaplanglt's state. This figure is adapted from~\cite[p.\ 52]{iris-technical-reference}.
    The additions for auxiliary state are shown in blue.
    The well-founded orders on levels ($\Levdom$) and degrees ($\Degdom$) can be modularly constructed (see main text, Section~\ref{sec:tclat-fair}).
  }%
  \label{fig:heaplanglt-state-appendix}
\end{figure}

\begin{figure}
  \input{figures/heaplanglt-pure.tex}
  \caption{\heaplang pure and boxed reduction rules. This figure is taken from \cite{iris-technical-reference}.}%
  \label{fig:heaplang-reduction-pure}
\end{figure}

\begin{figure}
  \input{figures/heaplang-impure-hred.tex}
  \caption{\heaplang's impure head reduction rules. Taken from Iris' technical documentation~\cite{iris-technical-reference}.}%
  \label{fig:heaplang-opsem}
\end{figure}

\begin{figure}
  \input{figures/heaplanglt-evaluation-contexts.tex}
  \caption{\heaplanglt evaluation contexts. The \langkw{let} construct is added but not auxiliary code.
    Auxiliary code added by \heaplanglt is highlighted in blue.
    }%
  \label{fig:heaplanglt-evalctxs}
\end{figure}

\begin{figure}
  \begin{mathpar}
    \let\prophnil\nil
    \newcommand{\inferC}[3]{\inferrule[#1]{#2}{#3}}
    \input{figures/heaplanglt-extensions.tex}
  \end{mathpar}
  \caption{Operations of \heaplanglt that extend or modify \heaplang.
    Reproduction of Figure~\ref{fig:heaplanglt-extensions} in the main text.}%
  \label{fig:heaplanglt-extensions-appendix}
\end{figure}

\begin{figure}
  \input{figures/heaplanglt-machine-steps.tex}
  \caption{Adapted from 
  iris' technical documentation~\cite{iris-technical-reference}.}%
  \label{fig:heaplanglt-machine-steps}
\end{figure}

\begin{figure}
\input{figures/heaplanglt-value-comparison}
\caption{\heaplang value comparison judgment. Copied from~\cite{iris-technical-reference}. }
\label{fig:heaplanglt-valeq}
\end{figure}

\section{Erasable \heaplanglt}%
\label{sec:heaplanglt-erasure}

\begin{figure}
  \let\prophnil\nil
  \begin{mathpar}

    \inferrule
      {
        \degree[0] \in \istate.\stateCallPerms(\theta) \\
        \aux{\expr_1, \istate \Downarrow_{\theta} \val', \istate'}
      }
      { ((\MyRecVI{\lvarF}{\lvar}{\aux{y}}{\expr_2})\spac \val\spac {\color{myghostcode} \expr_1}, \istate)
        \hstepi[\prophnil]{\theta}
        (\subst {\subst {\subst \expr \lvarF {(\MyRecEI{\lvarF}{\lvar}{\color{myghostcode}y}{\expr_2})}} \lvar \val} {\aux{y}} {\aux{\val'}}, \istate' : \stateCallPerms(\theta)\updsetminus\set{\degree[0]}, \nil) }

    \aux{\inferrule{
        \aux{\expr_2}, \istate \Downarrow_{\theta} \aux{\val}, \istate\aux{'} \\
        \aux{\subst{\expr_1}{\lvar}{\val}}, \istate\aux{'} \Downarrow_{\theta} \aux{\val'}, \istate\aux{''}
      }
      { (\aux{(\MyRecVA{\_}{\lvar}{\expr_1})\spac \expr_2}, \istate)
        \hstepi[\prophnil]{\theta}
        (\aux{\val'}, \istate\aux{''}, \nil) }}

    \inferrule
      { z > 0 \\
        \forall i < z.\, \istate.\stateHeap(\ell + i) = \bot \\
        \aux{\forall i < z.\, \istate.\stateHeapA(\ell + i) = \bot}
      }
      {(\AllocNI(z,v), \istate) \hstepi[\prophnil]{\theta} (\ell, \istate : \stateHeap[[\ell, \ell + z) \gets v] \aux{{}: \stateHeapA[[\ell, \ell + z) \gets 0]}, \nil)}

    \aux{\inferrule
      { z > 0 \\
        \forall i < z.\, \istate.\stateHeap(\ell + i) = \bot \\
        \aux{\forall i < z.\, \istate.\stateHeapA(\ell + i) = \bot}
      }
    {(\AllocNA(z,v), \istate) \hstepi[\prophnil]{\theta} (\ell, \istate : \stateHeapA[[\ell, \ell + z) \gets v], \nil)}}

    \inferrule
    {}
    { (\langkw{let}\spac x = v \spac\langkw{in}\spac \expr,\istate) \hstepi[\prophnil]{\theta}} (\expr[v/x],\istate,\nil)

    \aux{\inferrule
      {\expr_1, \istate \Downarrow_{\theta} v, \istate'}
      {(\LetA\spac {x} = {\expr_1}\spac \langkw{in}\spac {\expr_2}, \istate) \hstepi[\prophnil]{\theta} (\expr_2[v/x], \istate', \nil)}}

    \inferrule
      { \istate.\stateHeap(\ell) = \val }
      {(\deref\spac\ell,\istate) \hstepi[\prophnil]{\theta} (\val,\istate,\nil)}

    \aux{\inferrule
      { \istate.\stateHeapA(\ell) = \val }
      {(\derefA\spac\ell,\istate) \hstepi[\prophnil]{\theta} (\val,\istate,\nil)}}

    \inferrule
      { \istate.\stateHeap(\ell) = \val }
      {(\ell \getsR \valB,\istate) \hstepi[\prophnil]{\theta} (\TT,\istate : \stateHeap[\ell \gets \valB],\nil)}

    \aux{\inferrule
      { \istate.\stateHeapA(\ell) = \val }
      {(\ell \getsA \valB,\istate) \hstepi[\prophnil]{\theta} (\TT,\istate : \stateHeapA[\ell \gets \valB],\nil)}}
  \end{mathpar}
  \caption{Step rules of \heaplanglt extensions with adapted definitions compared to the main text to simplify the erasure proof.}%
  \label{fig:heaplanglt-erasable}
\end{figure}

To simplify the erasure argument, we adapt \heaplanglt in a few ways.
Figure~\ref{fig:heaplanglt-erasable} shows the most important changes.
The duplicated constructs let us easily distinguish between real and auxiliary code.

\subsection{Syntactic constraints}%
\label{sec:syntactic-constraints}
A safe \heaplanglt program satisfying the following syntactic constraints is erasable if:
\begin{itemize}
  \item The program is in A-normal form (ANF)~\cite{DBLP:conf/lfp/SabryF92}.
  \item The auxiliary let ($\LetA$) is the only auxiliary code embedded in real code.
    The other auxiliary constructs must appear in the val-body of the auxiliary let expression.
    Note that we consider atomic blocks to be real but that all our uses contain a single, real, atomic expression.
  \item We distinguish between real (instrumented) functions $\MyRecVI{\lvarF}{\lvar}{\aux{y}}{\expr_2}$ and purely auxiliary functions $\aux{\MyRecVA{\lvarF}{\lvar}{\expr}}$.
  \item An auxiliary let bound variable does not appear free in real code.
    It may, however, appear inside auxiliary let expressions inside real code.
    Moreover, our the real (instrumented) $\mu$ construct always takes a second argument where auxiliary let bound variables may appear.
  \item In the val-body of an auxiliary let, there are no assignments to the real heap i.e.\ \lstinline{:=} does not occur.
    The auxiliary assignment \lstinline{:==} may appear and gets stuck on real locations.
  \item Real functions may not be called in auxiliary code.
    This check is syntactic because real (instrumented) functions always take two arguments while auxiliary functions only take on. 
\end{itemize}

\subsection{Erasure}

\begin{figure}
  \begin{align*}
    \eraseop{\istate} = \{ \stateHeap \gets \istate.\stateHeap; \stateProphs \gets \istate.\stateProphs \}
  \end{align*}
  \caption{Store erasure, i.e.\ removing the auxiliary state of \heaplanglt.}%
  \label{fig:erasure-store}
\end{figure}

\begin{figure}
  \begin{align*}
    \eraseop{T;\istate} = \mathsf{map}\spac\mathsf{\eraseop}\spac T; \eraseop{\istate}
  \end{align*}
  \caption{A machine configuration consisting of a list expressions, one per thread, $T$ and the state $\istate$ erases to the erasure of the every expression and the erased state.}%
  \label{fig:erasure-configuration}
\end{figure}

\begin{figure}
  \begin{align*}
    \eraseop{\aux{\LetA\spac {x} = {\expr_1} \langkw{in}} \spac {\expr_2}} = \eraseop{\expr_2} \\
    \eraseop{\MyRecEI{\lvarF}{\lvar}{\aux{y}}{\expr}} = \MyRecVI{\lvarF}{\lvar}{\aux{\TT}}{\eraseop{\expr}} \\
    \eraseop{\langkw{AtomicBlock}\spac \expr} = \langkw{AtomicBlock}\spac \eraseop{\expr_1} \\
    \eraseop{(\expr_1\spac \expr_2\spac \expr_3)} = (\eraseop{\expr_1}\spac \eraseop{\expr_2}\spac \aux{()}) \\
    \text{Other cases by straightforward structural recursion, e.g.:} \\
    \eraseop{\LetR{x}{\expr_1}{\expr_2}} = \LetR{x}{\eraseop{\expr_1}}{\eraseop{\expr_2}}
  \end{align*}
  \caption{Expression erasure.}%
  \label{fig:erasure-expression}
\end{figure}

\begin{theorem}[Erasure]\label{thm:erasure-appendix}
  $\eraseop{e'}$ does not admit infinite fair executions if the augmented program $e';\Finish$ is safe.
\end{theorem}
\begin{proof}

Assume to the contrary an arbitrary infinite fair execution of $\eraseop{e'}$.
In contradiction of Theorem~\ref{thm:no-infinite-fair-executions} we construct an infinite fair execution of $e';\langkw{Finish}$.

\begin{lemma}[Augmented infinite fair execution]\label{lem:matching-execution}
  Let $T';\istate'$ be an arbitrary, safe, augmented configuration and $\Sigma$  a fair infinite execution starting from $\eraseop{T'};\eraseop{\istate'}$.
  Then, there exists an infinite fair execution starting from $T';\istate'$.
\end{lemma}
\begin{proof}
Let $T$ be the erased list of expression representing the threads and $i$ be the index in $T$ of the thread taking the first step of $\Sigma$.
We compare the head redexes of $T'[i]$ and $T[i]$ before that step.

\textbf{Base case}:\\
If they are the same, we execute the real step in the augmented execution as well.
Notice that real code does not interfere with auxiliary state so that the we can apply the same logic to the following step.
Thus, we apply the lemma coinductively on the resulting two new configurations.

\textbf{Inductive case}:\\
Otherwise, $T[\theta]$ and $T'[\theta]$ do not correspond.
$T'[\theta]$ has the shape of finitely many nested auxiliary let statements before reaching the head redex of $T[\theta]$.
By the safety of $T';\istate'$ we know that all auxiliary steps are safe and terminating.
The finite number of lets are executed until the base case is reached.
\end{proof}

Recall that we assume by contradiction that $[\eraseop{e'}],\emptyset$ (i.e.\ starting from an empty heap) is an infinite fair execution.
Moreover, we know that $e'$ is safe given some state $\istate'$ (with $\istate'.\stateHeap = \emptyset$) satisfying the weakest precondition of $e';\langkw{Finish}$ i.e.\ equipped with sufficient initial call permissions.
By definition $\eraseop{[e';\Finish],\istate'} = [\eraseop{e'}],\emptyset$.
Applying Lemma~\ref{lem:matching-execution} we obtain an infinite fair execution, for $e';\langkw{Finish}$, reaching the desired contradiction.
\end{proof}

\end{document}